\documentclass{llncs}

\usepackage[utf8]{inputenc}
\usepackage{amsfonts, amsmath, amssymb}
\usepackage[dvipsnames]{xcolor}
\usepackage{url}
\usepackage{graphicx}

\usepackage{tikz}
\usetikzlibrary{quantikz2}
\usetikzlibrary{decorations.pathreplacing}
\definecolor{duckBlue}{HTML}{56949f}
\definecolor{goldOrange}{HTML}{907AA9}
\definecolor{love}{HTML}{B4637A}
\usepackage[hidelinks]{hyperref}
\hypersetup{%
  bookmarks=true,
  colorlinks = true,
  linkcolor  = violet,          
  citecolor  = RedOrange,          
  filecolor  = violet,          
  urlcolor   = TealBlue
}

\usepackage{algorithm}
\usepackage{algorithmicx}
\usepackage[noend]{algpseudocode}
\usepackage{comment}
\usepackage{bm}
\usepackage{braket}
\usepackage{cancel}

\spnewtheorem{assumption}{Assumption}{\itshape}{\rmfamily}

\renewcommand*{\L}{\mathcal L}
\newcommand*{\A}{\mathcal A}

\newcommand*{\Z}{\mathbb Z}

\newcommand*{\R}{\mathbb R}

\newcommand*{\N}{\mathbb N}

\newcommand*{\poly}{\mathrm{poly}}

\renewcommand*{\exp}[1]{e^{#1}}
\renewcommand*{\vec}[1]{\bm{#1}}
\newcommand*{\polylog}[1]{\mathrm{polylog}\!\left(#1\right)}
\newcommand{\SIS}{\mathrm{SIS}}
\newcommand{\mat}[1]{\mathbf{#1}}

\DeclareMathOperator{\spann}{span}

\newcommand{\clem}[1]{\textcolor{Plum}{\textbf{Clem:} #1}}
\newcommand{\yixin}[1]{\textcolor{blue}{\textbf{Yixin:} #1}}
\newcommand{\as}[1]{\textcolor{Red}{\textbf{André: } #1}}

\renewcommand{\clem}[1]{}
\renewcommand{\yixin}[1]{}
\renewcommand{\as}[1]{}

\newif\ifeprint
\eprinttrue

\usepackage{thm-restate}

\algdef{SE}[DOWHILE]{Do}{doWhile}{\algorithmicdo}[1]{\algorithmicwhile\ #1}%

\makeatletter
\AtBeginDocument{%
  \def\doi#1{\url{https://doi.org/#1}}}
\makeatother

\usepackage{mathtools, stmaryrd}
\usepackage{xparse} \DeclarePairedDelimiterX{\iintv}[1]{\llbracket}{\rrbracket}{\iintvargs{#1}}
\NewDocumentCommand{\iintvargs}{>{\SplitArgument{1}{,}}m}
{\iintvargsaux#1} %
\NewDocumentCommand{\iintvargsaux}{mm}{#1 , #2}

\title{Quantum algorithm for Discrete Gaussian Sampling}
\ifeprint
\author{Clémence Chevignard \and Yixin Shen \and André Schrottenloher}
\institute{Univ Rennes, Inria, CNRS, IRISA, Rennes, France}
\else
\author{}
\institute{}
\fi

\newcommand{\osguess}{\overline{\vec{s}}_{\mathrm{guess}}}

\newcommand{\sguess}{\vec{s}_{\mathrm{guess}}}
\newcommand{\sdual}{\vec{s}_{\mathrm{dual}}}
\newcommand{\Aguess}{{A}_{\mathrm{guess}}}
\newcommand{\Adual}{{A}_{\mathrm{dual}}}
\newcommand{\nguess}{n_{\mathrm{guess}}}
\newcommand{\ndual}{n_{\mathrm{dual}}}
\newcommand{\floor}[1]{\lfloor #1 \rfloor}

\newcommand{\bigO}[1]{O\! \left(#1\right) }
\newcommand{\bigOt}[1]{\widetilde{O}\! \left(#1\right) }

\begin{document}

\renewcommand{\sectionautorefname}{Section}
\renewcommand{\subsectionautorefname}{Section}
\def\algorithmautorefname{Algorithm}

\maketitle


\begin{abstract}
Discrete Gaussian Sampling on lattices is a fundamental problem in lattice-based cryptography. It appears both in basic cryptographic primitives such as digital signatures and as an important cryptanalysis building block for solving hard lattice problems. In this paper, we show a quantum algorithm based on the quantum rejection sampling technique whose complexity is asymptotically quadratically faster than its classical counterpart in [Wang \& Ling, IEEE Trans. Inf. Theory 2019]. Our sampler outputs a quantum state which can either be measured to get the desired distribution or be used directly as such in other quantum algorithms. By doing so, we derive two versions of quantum dual attacks that improve upon the previous ones in [Pouly \& Shen, EUROCRYPT 2024]. The two versions are incomparable, each having distinct advantages (speed vs memory requirement). The second version is particularly interesting as it requires only polynomial classical and quantum memory, excluding the classical memory used in the preprocessing step of the Discrete Gaussian sampler. Our quantum Discrete Gaussian sampler can also be used to speed up the algorithm for solving the Short Integer Solution problem, in any norm, of [Bollauf, Pouly \& Shen, ePrint 2026/225].
\end{abstract}


\section{Introduction}

Discrete Gaussian Sampling on lattices is a fundamental problem in lattice-based cryptography. Given a lattice, a vector $\vec c$ and a parameter $\sigma$, it asks to sample a lattice vector $\vec x$ with probability proportional to $\exp{-\frac{\pi \|\vec x - \vec c \|^2}{\sigma^2}}$. The problem arises both in basic cryptographic primitives such as ``hash-and-sign'' digital signature schemes~\cite{DBLP:conf/stoc/GentryPV08,Falcon}, and in cryptanalysis, where it serves as a fundamental tool for solving hard problems like the Shortest Vector Problem~\cite{DBLP:conf/stoc/AggarwalDRS15}, the Learning with Errors problem \cite{DBLP:conf/eurocrypt/PoulyS24,DBLP:conf/asiacrypt/QuX25} and the Short Integer Solution problem \cite{PS26}. 

The width of the distribution is controlled by $\sigma$. The smaller the width, the harder it is to sample from this distribution. A particularly important threshold is known
as the smoothing parameter \cite{DBLP:conf/focs/MicciancioR04} which marks the transition between the distribution looking more like a continuous
Gaussian, and the distribution being more heavily concentrated on the closest lattice vector to the center $\vec{c}$.

There are several algorithms to sample from the discrete Gaussian distribution with various trade-offs.
The most well-known algorithm, known as the Klein sampler, is very efficient (polynomial time) but the width depends on the quality of the input basis~\cite{DBLP:conf/stoc/GentryPV08}. This makes this algorithm useful for relatively large values of $\sigma$, much larger than the smoothing parameter, such as in cryptographic primitives.
At the other end of the spectrum is the algorithm by \cite{DBLP:conf/stoc/AggarwalDRS15} which can sample the (centered) discrete Gaussian for
any value of $\sigma$ in time $2^{n+o(n)}$ and above the smoothing parameter in time $2^{n/2+o(n)}$. While this algorithm has some theoretical applications \cite{DBLP:conf/stacs/AggarwalCKS21,DBLP:journals/siamcomp/AggarwalCKS25}, its complexity is too high for practical attacks on cryptosystems.
The algorithm of~\cite{DBLP:conf/stacs/AggarwalCKS21} provides a time-space tradeoff ranging from exponential to super-exponential time, and exponential to polynomial space. This algorithm only works for a width above the smoothing parameter.
Finally, the algorithms of \cite{MTMK} and \cite{DBLP:conf/stoc/BrakerskiLPRS13}, which are essentially equivalent despite using a very different
presentation\footnote{See Remark 1 in \cite{DBLP:conf/pqcrypto/PoulyS25}.}, work for any $\sigma$ and provide an interesting trade-off where the complexity depends on both the quality of the input basis and $\sigma$. This algorithm has been used in several attacks \cite{DBLP:conf/eurocrypt/PoulyS24,DBLP:conf/pqcrypto/PoulyS25,DBLP:conf/asiacrypt/QuX25,PS26} in the regime where its complexity is exponential
but much smaller than $2^{n}$.

There is no known quantum speedup on the above algorithms. The Klein sampler, which only requires polynomial space, can
easily be turned into a quantum sampler (as shown in~\cite{DBLP:journals/iacr/BoerF25}, which we reuse in~\autoref{sec: quantum ways of samplig from pi}) but does not provide a speed-up per se and
still suffers from the strong constraint on the width. The algorithm of~\cite{DBLP:conf/stoc/AggarwalDRS15} requires exponential memory and cannot be turned into a quantum sampler without using quantum RAM. Finally, the algorithm of \cite{MTMK,DBLP:conf/stoc/BrakerskiLPRS13} can also be turned into a quantum sampler with a polynomial number of qubits with no particular speedup.

\subsubsection*{Contributions.}

In this paper, we present a quantum algorithm allowing to sample vectors according to a discrete  Gaussian distribution, with an (asymptotically) quadratic speed-up on the complexity of the algorithm of~\cite{MTMK}. Our sampler works for any width $\sigma$ and its complexity depends on the width and the quality of the input basis.

Our algorithm is based on \emph{quantum rejection sampling}~\cite{DBLP:journals/toct/OzolsRR13}, i.e., transforming the amplitudes of a quantum state from a given distribution to a target distribution -- in a way analogous to classical rejection sampling. Consider that we want to sample vectors in a lattice $\L$, according to a Gaussian lattice distribution $D$, then our objective is to build a quantum state (close to) $\sum_{\vec x\in \L}\sqrt{D(\vec x)}\ket{\vec  x}$, since measuring this state gives us a vector $\vec x$ sampled according to $D$.

In order to build this target state, we start from an algorithm which prepares another state $\sum_{\vec  x\in \L}\sqrt{D_1(\vec x)}\ket{\vec  x}$, which in our case is a quantum version of Klein's sampling algorithm, previously given in~\cite{DBLP:journals/iacr/BoerF25}. We modify the target amplitudes, and then use quantum amplitude amplification over these steps to reach the target state.

\paragraph{Application to SIS and Dual Attacks on LWE.} 
We use our quantum sampling algorithm in three distinct algorithms solving the Learning with Errors (LWE) and the Short Integer Solution (SIS) problems.
Our first application consists in replacing the first step of the quantum dual attack of Pouly and Shen~\cite{DBLP:conf/eurocrypt/PoulyS24}, which asks to sample vectors according to a discrete Gaussian. Using our algorithm, we can obtain a quadratic speedup on this first step (excluding the preprocessing step for BKZ reduction), while the second step (which uses the data produced inside an FFT-based distinguisher) remains unchanged.

Our second application is a new quantum dual attack, in which the classical data used in the second step of the previous quantum dual attack~\cite{DBLP:conf/eurocrypt/PoulyS24} is replaced by our quantum sampling algorithm. This algorithm has a higher time complexity, but uses only polynomial quantum memory.

Our third application consists in quantizing the algorithm of Bollauf, Pouly and Shen~\cite{PS26} for solving the SIS problem for any norm. Using our quantum sampling algorithm and quantum amplitude amplification, we can construct a quantum state that provides a solution of the SIS problem with high probability. Excluding the BKZ basis reduction preprocessing step, which itself allows for a quantum speedup, the remaining step achieves a quadratic speedup over its classical counterpart. Ignoring polynomial factors in the complexity, we provide rough estimates of the quantum complexity of the $\SIS^\infty$  problem in the Dilithium signature scheme.

Like many works in lattice-based quantum cryptanalysis~\cite{DBLP:conf/asiacrypt/ChaillouxL21,DBLP:journals/iacr/AlbrechtS22,DBLP:conf/eurocrypt/BonnetainCSS23}, we focus on the asymptotic complexity of algorithms and on theoretical quantum speedups against classical algorithms. It has been acknowledged before that such asymptotic speedups may be difficult to ever realize in practice, especially if they require memory, due to the inherent cost in maintaining quantum RAM and the sequentiality of Grover's quantum search~\cite{DBLP:conf/asiacrypt/AlbrechtGPS20}. While these issues would also present themselves for the algorithms that we study in this paper, our goal is primarily to expand the quantum cryptanalysis toolbox. We leave detailed considerations for implementations, such as depth and non-asymptotic gate counts, for future work.

\paragraph{Organization of the Paper.} In~\autoref{sec:prelim}, we present several useful notions: lattices and Gaussian distributions on lattices, the LWE and the SIS problems, and some useful quantum algorithms. In~\autoref{sec: quantum ways of samplig from pi}, we describe our quantum algorithm for sampling from discrete Gaussian distributions: \autoref{sec:comparison} shows that this algorithm essentially obtains a quadratic speedup compared to the Markov Chain Monte Carlo (MCMC) algorithm~\cite{MTMK}.

We then move to applications. In~\autoref{sec: Improving the quantum dual attack}, we apply our quantum sampling algorithm to the ``modern'' dual attack framework on LWE of~\cite{DBLP:conf/eurocrypt/PoulyS24}; first by introducing the quantum sampler in the first step of the algorithm, second by introducing the quantum sampler in the second step. In~\autoref{sec: Improving SIS}, we apply our quantum sampling algorithm to quantize the algorithm of \cite{PS26} for solving the SIS problem. We give some rough estimates on the quantum complexity of $\SIS^\infty$ used in the Dilithium signature scheme. 

\section{Preliminaries}\label{sec:prelim}


Throughout this paper, we assume familiarity with basic quantum computing notions such as quantum states, qubits, the ket notation $\ket{\cdot}$, etc. We write quantum algorithms as quantum circuits, which is the most common approach in quantum cryptanalysis; a full overview of the quantum circuit model can be found in~\cite{nielsen2010quantum}. The gate count (i.e., number of operations) is our main metric of interest, but since we consider asymptotic algorithms, we do not specify explicitly a universal gate set.

The algorithms that we consider have polynomial qubit count, but they may require exponential-size  \emph{quantum RAM} (qRAM). The model used is \emph{classical memory with quantum random access} (QRACM). It extends the quantum circuit model with the operation:
$\ket{i} \ket{x} \mapsto \ket{i} \ket{x \oplus m_i} \enspace,$
where $(m_0, \ldots, m_{\ell-1})$ is an exponential-size classical memory and $i$ is an index register. When using QRACM, we make the assumption that a specific physical component allows to implement this operation in time $\bigO{1}$, since the circuit model cannot. qRAM has been part of all quantum speedups for lattice sieving to date~\cite{DBLP:conf/pqcrypto/LaarhovenMP13,DBLP:conf/asiacrypt/ChaillouxL21,DBLP:conf/eurocrypt/BonnetainCSS23}. A recent paper~\cite{DBLP:journals/iacr/ChoHKLS24} showed that under some reasonable hypothesis, there is no quantum speed-up possible for lattice sieving without using qRAM. This is opposed to lattice enumeration~\cite{DBLP:conf/asiacrypt/AonoNS18}, since it runs with polynomial memory. Indeed, any quantum algorithm using polynomial space overall does not need qRAM.

For $\Omega$ a set and $\alpha_j > 0$, the state $\ket{\psi} = \sum_{j \in \Omega} \alpha_j \ket{j}$ represents a probability distribution $\Pi$ over $\Omega$, defined by $\Pi(j) = \alpha_j^2$. When this is clear from context, we will denote this state as $\ket{\Pi}$. To measure the distance between two probability distributions $\mu_1$ and $\mu_2$, we use the \emph{total variation distance}, defined as
\begin{equation*}
    d_{TV}(\mu_1,\mu_2) = \frac{1}{2} \sum_{x\in \Omega}|\mu_1(x)-\mu_2(x)| \enspace.
\end{equation*}

\subsection{Lattices}

A $\R$-\emph{lattice} $\L$ is a discrete subgroup of $\R^n$ spanned by a set of independent vectors $\vec b_1,\dotsc,\vec b_n\in \R^m$ over $\Z$, meaning that $\L = \{\sum_{i=1}^n a_i \vec b_i, a_i\in \Z\}$. When $n=m$, we say that the lattice $\L$ is full rank. The vectors $\vec b_1,\dotsc, \vec b_n$ form a \emph{basis} of $\L$. Usually, to represent this basis, we concatenate these column vectors into a matrix $B$, and we write $\L = \L(B)$. The matrix $\widetilde B$ denotes the Gram-Schmidt decomposition of $B$, and we denote by $\widetilde{\vec b}_i$ its column vectors. The length of the shortest nonzero vector of $\L(B)$ is denoted by $\lambda_1(\L(B))$.

The \emph{fundamental parallelepiped} of $\L(B)$ is $P(\L(B)) := \{ \sum_{i=1}^n a_i \vec b_i, a_i\in [0,1)\}$. The volume of the lattice $\L(B)$ is the volume $\mathrm{vol}(P(\L(B)))$ of $P(\L(B))$. It also equals $\sqrt{\det(B^tB)}$, or simply $|\det(B)|$ when $B$ is square and full rank.

Define, in the same fashion, $\spann_\R(\L(B)) := \{\sum_{i=1}^n a_i \vec b_i, a_i\in \R\}$. The dual $\hat \L$ of the lattice $\L$ is defined as 
\begin{equation*}
    \hat \L := \{\vec y \in \spann_\R(\L) \text{ such that }\forall \vec x \in \L, \langle \vec x,\vec y\rangle \in \Z\} \enspace.
\end{equation*}

Let $1\leqslant n \leqslant m$, and let $q$ be a prime power. A lattice $\L$ is a $m$-dimensional $q$-ary lattice if $q\Z^m \subset \L \subset \Z^m$. Given a matrix $A\in \Z^{m\times n}$, we can define two $m$-dimensional $q$-ary lattices
\begin{gather*}
    \L_q(A) := \{\vec x\in \Z^m \text{ such that }\exists \vec s\in \Z^n \text{ such that } A \vec s = \vec x\bmod{q} \},\\
    \L_q^\perp(A) := \{\vec x\in \Z^m \text{ such that } A^T \vec x = 0\bmod{q} \}.
\end{gather*}
There exists matrices $A$ and $B$ such that $\L = \L_q(A) = \L_q^\perp(B)$. Moreover, $\widehat{\L_q^\perp(A)} = \frac{1}{q}\L_q(A)$, and $\det(\L_q(A)) = q^{m-rk(A)}$, $\det(\L_q^\perp(A)) = q^{rk(A)}$. 
From now on, we will always consider that $A$ has full rank.

\subsection{Discrete Gaussian Distributions on Lattices and Variants}\label{sec: Lattice Gaussian Distribution}

\subsubsection{Discrete Gaussian Distributions.}\label{sec:gaussian lattice distrib}

Consider vectors $\vec c, \vec x\in \R^m$ and a real $\sigma>0$. The \emph{Gaussian function centered at $\vec c$} for standard deviation $\sigma$ is 
\begin{equation*}
    \rho_{\sigma,\vec c}(\vec x) := \exp{-\frac{\pi||\vec x - \vec c||^2}{\sigma^2}} \enspace.
\end{equation*}
From this function, we define \emph{Discrete Gaussian distributions} on lattices. Let us fix $\vec c\in \R^m$, and consider a lattice $\L\in \R^m$, some $\vec x\in \L$, and $\sigma>0$. The Gaussian lattice distribution defined on $\L$, centered in $\vec c$ and with standard deviation $\sigma$, is 
\begin{equation*}
    D_{\L, \sigma, \vec c}(\vec x) := \frac{\rho_{\sigma,\vec c}(\vec x)}{\rho_{\sigma,\vec c}(\L)} = \frac{\exp{-\frac{\pi}{\sigma^2}||\vec x -\vec c||^2}}{\sum_{\vec y\in \L} \exp{-\frac{\pi}{\sigma^2}||\vec y -\vec c||^2}} \enspace.
\end{equation*}
The denominator $\rho_{\sigma,\vec c}(\L)$ is a scaling factor so that $D_{\L, \sigma, \vec c}(\L) = 1$. When $\vec c = 0$, we simply write $D_{\L, \sigma}$ instead of $D_{\L, \sigma, \vec c}$. We use the following lemmas to bound Gaussian masses and their tails.


\begin{lemma}[{\cite[Lemma 2.9]{DBLP:journals/siamcomp/MicciancioR07}}]\label{lemma:rho sigma c}
    Let $\L$ be a lattice in $\R^m$. Then, for any $\sigma>0$ and $\vec c\in \R^m$, we have 
    $\rho_{\sigma,\vec c}(\L)\leqslant \rho_{\sigma}(\L)$.
\end{lemma}

\begin{restatable}{lemma}{rhoZapprox}[From Lemma 2 in~\cite{DBLP:conf/pqcrypto/PoulyS25}]\label{lemma:rho Z approx}
For any $\sigma > 0$, $\rho_\sigma(\Z) \leq 3\sigma + 4 \enspace.$
\end{restatable}

\begin{restatable}{lemma}{gaussianfinitesets}\label{lemma: inequality on gaussian masses on finite sets}
Consider an integer $M\in \N^*$, and real numbers $\sigma>0$ and $\vec c\in \R$. Then, we have
    \begin{gather*}
        \rho_{\sigma,\vec c}(\Z)\leqslant \frac{\rho_{\sigma}(\iintv{-2^M,2^M})}{1-2^M\sqrt{2\pi e}\times \exp{-\pi 2^{2M}}} \text{ and } \rho_{\sigma,\vec c}(\Z\backslash\iintv{-2^M,2^M})\leqslant O(\sigma\times 2^M \exp{-\pi 2^{2M}}) \enspace.
    \end{gather*}
\end{restatable}

The proof is given in Appendix~\ref{sec: Proof of some results}.

%
%
%

\subsubsection{Periodic Gaussian Function.}
The periodic Gaussian function is defined as follows: 
\begin{equation}\label{eq:periodic Gaussian function}
f_{\L,\sigma}(\vec x\in \R^m) := \frac{\rho_\sigma(\L+\vec x)}{\rho_\sigma(\L)}
\end{equation}

This function is clearly periodic over $\L$. Therefore, one can define its Fourier coefficients as follows:
\begin{equation*}
\text{For }w\in \hat{\L},\ \hat{f}_{\L,\sigma}(w) = \frac{1}{\mathrm{vol}(\L)} \int_{z\in P(\L)} f_{\L,\sigma}(z)\exp{-2\pi i \langle w,z\rangle}dz \enspace.
\end{equation*}
Moreover, the Fourier series of $f_{\L,\sigma}$ at $\vec{x}$ is defined as:
\begin{equation*}
    \sum_{w\in \hat{\L}}\hat{f}_{\L,\sigma}(w)\exp{2\pi i \langle \vec x,\vec w\rangle} \enspace.
\end{equation*}

\begin{lemma}[{\cite[Fact 2.3]{DBLP:journals/jacm/AharonovR05}}]\label{lemma:fourier serie equal to f}
    For any sufficiently smooth function $f:\R^n \to \R$ that is periodic over some lattice $\L$ and any $x\in \R^n$, the Fourier series of $f$ at $x$ is equal to $f(x)$.
\end{lemma}
The function $f_{\L,\sigma}$ is sufficiently smooth for this lemma to apply, for any $\L(B)$ and $\sigma>0$. We then have the following result.
\begin{lemma}[{\cite[Claim 4.1]{DBLP:journals/jacm/AharonovR05}}]\label{lemma:Fourier serie of f L,sigma}
    For any lattice $\L\subset \R^m$, and $\sigma>0$, we have $\hat f_{\L,\sigma} = D_{\hat \L,1/\sigma}$, which is a probability measure over the dual lattice $\hat \L$.
\end{lemma}

Later on, we will need to estimate the distance of an element to a lattice $\L$ efficiently. One of the results used in \cite{DBLP:conf/eurocrypt/PoulyS24} to do so is the following Lemma, which is based on the Chernoff-Hoeffding bound.

\begin{lemma}[Pointwise approximation {\cite[Lemma 1.3]{DBLP:journals/jacm/AharonovR05},\cite[Lemma 6]{DBLP:conf/eurocrypt/PoulyS24}}]\label{theorem: Poitwise approximation lemma}
    Consider $\L$ a full-rank lattice in $\R^m$, and $h: \R^m \to \R$ that is periodic over $\L$ and whose Fourier series $\hat h$ is a probability measure over $\hat \L$. Let $N$ be an integer, $\delta>0$, and $X\subset \R^m$ a finite set. Let $W = (\vec w_1,\dotsc,\vec w_N)$ be a list of vectors in the dual lattice chosen randomly and independently from the distribution $\hat h$. Then, with probability at least $1-|X|2^{-\Omega(N\delta^2)}$, for all $\vec x \in \L+X$:
    \begin{equation*}
    |h_W(\vec x) - h(x)|\leqslant \delta,\text{ where } h_W(\vec x) := \frac{1}{N}\sum_{i=1}^N \cos(2\pi\langle \vec w_i,\vec x\rangle )\enspace .
    \end{equation*}
\end{lemma}

By~\autoref{lemma:Fourier serie of f L,sigma}, the pointwise approximation lemma applies to $f_{\L,\sigma}$. 

\subsubsection{Variants of Discrete Gaussian Distributions for Finite Probability Spaces.}\label{sec: Variants of Gaussian lattice distributions for finite probability spaces}

In the rest of the paper, we will build registers $\ket s$ representing probability distributions close to either $D_{\L, \sigma}$, or the ``Klein distribution'' (that we will define later). 
However, a register contains only finitely many qubits, so it can only represent a probability distribution over a finite space, while $D_{\L, \sigma}$ is defined over an infinite space. Let $\Omega$ be a finite subset of $\L$, we define the \emph{finite discrete gaussian distribution} with standard deviation $\sigma$ as
\begin{equation*}
    D_{\Omega, \sigma}(\vec x\in \Omega) = \frac{\rho_{\sigma}(\vec x)}{\rho_{\sigma}(\Omega)} = \frac{\exp{-\frac{\pi}{\sigma^2}||\vec x||^2}}{\sum_{\vec y\in \Omega} \exp{-\frac{\pi}{\sigma^2}||\vec y||^2}} \enspace.
\end{equation*}

We have:
\begin{align}\label{eq: dtv of gaussian distrib over finite space vs infinite space}
    2d_{TV}(D_{\Omega, \sigma},D_{\L, \sigma}) &= \sum_{\vec x\in \Omega}|D_{\Omega, \sigma}(\vec x) - D_{\L, \sigma}(\vec x)| + \sum_{\vec x\in \L\backslash \Omega}|D_{\Omega, \sigma}(\vec x) - D_{\L, \sigma}(\vec x)|\nonumber \\
    &= \sum_{\vec x\in \Omega}D_{\Omega, \sigma}(\vec x)\left(1-\frac{\rho_{\sigma}(\Omega)}{\rho_{\sigma}(\L)}\right) + \sum_{\vec x\in \L\backslash \Omega}D_{\L, \sigma}(\vec x)\nonumber \\
    &= \left(1-\frac{\rho_{\sigma}(\Omega)}{\rho_{\sigma}(\L)}\right) + \rho_{\sigma}(\L\backslash \Omega)\nonumber\\
    &= \rho_{\sigma}(\L\backslash \Omega)(1+1/\rho_{\sigma}(\L))\leqslant 2\rho_{\sigma}(\L\backslash \Omega) \enspace.
\end{align}

\begin{restatable}{proposition}{somedtvgaussianmass} \label{prop: bound on some dtv and gaussian mass}
    Consider a full rank lattice $\L(B)\in \R^m$, along with a finite subset $\Omega \supset B\prod_{i=1}^m \iintv{-2^{M-1},2^{M-1}}$, for some integer $M>0$. 
    We have:
    \begin{equation*}
    \rho_{\sigma}(\L\backslash \Omega) \leqslant 2 \left(2^{M-1}\sqrt{2\pi e}\right)^m\times \exp{-\pi\cdot m\cdot 2^{2M-2}}  (3\sigma / \min_i \| \widetilde{\vec b}_i\| + 4)^m \enspace.
    \end{equation*}
\end{restatable}

The proof is given in Appendix~\ref{sec: Proof of some results}.

Now, we define the (finite) \emph{Klein distribution} which arises from a variant of the Klein sampling algorithm~\cite{DBLP:conf/soda/Klein00}. Indeed, the probability space of the original Klein distribution of~\cite{DBLP:conf/soda/Klein00} is $\L(B)$, which is an infinite space. Therefore we define a variant that samples from a finite set $\Omega$ of $\L(B)$.

\begin{algorithm}[htbp]
    \caption{Klein's sampling algorithm variant}\label{alg:Klein's algorithm}
\begin{algorithmic}[1]
\Statex \textbf{Input:} $B$, $\sigma$, $\vec c$, and an integer $M>1$
\smallskip
\Statex  \textbf{Output: } $B\vec x \in \Omega$
\medskip
\State Write $B$ as $B = QR$, with $Q$ orthogonal and $R$ upper triangular, with $|R_{i,i}| = ||\widetilde{\vec{b}_i}||$
\State $\vec c' \leftarrow Q^{-1}\vec c$
\For{$i=m \to 1$}
	\State $\sigma_i\leftarrow \sigma/|R_{i,i}|$
    \State $\widetilde x_i \leftarrow \frac{c'_i-\sum_{j=i+1}^mR_{i,j}x_j}{R_{i,i}}$
    \State Sample $x_i$ from $D_{\iintv{-2^M+\lfloor \widetilde x_i \rceil ,2^M+\lfloor \widetilde x_i \rceil},\sigma_i,\widetilde x_i}$
\EndFor
\State \textbf{Return} $B\vec x$
\end{algorithmic}
\end{algorithm}

Here $\vec c$ is, again, the center of the distribution. Let us denote by $q$ the Klein distribution yielded by~\autoref{alg:Klein's algorithm}. We have, since $Q$ is orthogonal:
\begingroup \allowdisplaybreaks
\begin{align}
    \label{equation:q expression}
    q&(B\vec x) = \prod_{i=1}^m \mathbb P(x_{m+1-i}|(x_{m-i},\dotsc,x_1)) = \prod_{i=m}^1 \frac{\rho_{\sigma_{i},\widetilde{x}_{i}}(x_{i})}{\rho_{\sigma_{i},\widetilde{x}_{i}}(\iintv{-2^M+\lfloor \widetilde x_i \rceil,2^M+\lfloor \widetilde x_i \rceil})}\nonumber \\
    &= \prod_{i=m}^1 \frac{\exp{-\frac{\pi}{\sigma^2}((Q^{-1}\vec c)_{i}-\sum_{j=i+1}^mR_{i,j}x_j-x_{i}R_{i,i})^2}}{\rho_{\sigma_{i},\widetilde{x}_{i}}(\iintv{-2^M+\lfloor \widetilde x_i \rceil,2^M+\lfloor \widetilde x_i \rceil})} = \frac{\exp{-\frac{\pi}{\sigma^2} \sum_{i=m}^1\left((Q^{-1}\vec c)_{i}-\sum_{j=i}^mR_{i,j}x_j\right)^2}}{ \displaystyle\prod_{i=m}^1 \rho_{\sigma_{i},\widetilde{x}_{i}}(\iintv{-2^M \! + \! \lfloor \widetilde x_i \rceil,2^M \! + \! \lfloor \widetilde x_i \rceil})}\nonumber \\
    &= \frac{\rho_{\sigma,\vec c}(B\vec x)}{ \displaystyle\prod_{i=m}^1 \rho_{\sigma_{i},\widetilde{x}_{i}}(\iintv{-2^M \! + \! \lfloor \widetilde x_i \rceil,2^M \! + \! \lfloor \widetilde x_i \rceil})}
\end{align}
\endgroup

The probability space $\Omega$ is a bit hard to fully describe. One can however note that it contains the space $B\times \prod_{i=1}^m\iintv{-2^{M-1},2^{M-1}}$, which, when $M$ is chosen large enough, ensures that $D_{\Omega, \sigma}$ is close to $D_{\L, \sigma}$. Later (in the proof of~\autoref{theorem: Klein distrib}) we give an upper bound on the size of the vectors sampled by the algorithm.

\subsection{LWE and SIS}\label{sec:LWE-SIS}

The LWE problem was first introduced by Regev in~\cite{DBLP:journals/jacm/Regev09}. Recent NIST standardized lattice-based public key cryptosystems, Kyber~\cite{DBLP:conf/eurosp/BosDKLLSSSS18} and Dilithium~\cite{DBLP:journals/tches/DucasKLLSSS18}, both rely on the classical and quantum computational hardness of (module version of) LWE.

Let $n,m,q\in \N$, and let $\chi_e$ be a probability distribution over $\Z_q$, which is called the \emph{noise distribution}. For any vector $\vec s\in \Z_q^n$, we denote by $LWE(m,\vec s,\chi_e)$ the probability distribution obtained by sampling a matrix $A\in \Z_q^{m\times n}$ uniformly at random, sampling a vector $\vec e\in \Z_q^m$ according to $\chi_e^m$, and outputting $(A,\vec b := A\vec s + \vec e)$. We denote by $A_i$ the $i$-th row of $A$, $e_i$ the $i$-th coefficient of $e$, and so $b_i = \langle A_i,\vec s\rangle + \vec e_i\in \Z_q$. The search-LWE problem is, given oracle access to a sampler $LWE(m,\vec s,\chi_e)$, to find $\vec s$.

There exists two main strategies to solve LWE: the \emph{primal attack} and the \emph{dual attack}. The latter has been studied both in the classical and quantum setting, and will be detailed in~\autoref{sec: Improving the quantum dual attack}.


The short integer solution (SIS) problem is an important problem in lattice-based cryptography. It asks to find a short non-zero vector in a randomly generated 
$q$-ary lattice of a certain form under the $\ell_p$ norm.

\begin{definition}\label{def:sis}
	Let $m\geqslant n\geqslant1$, $q$ be an integer, $\ell>0$ and $p\in(0,\infty]$.
	The search problem $\SIS_{n,m,q,\ell}^p$ (Short Integer Solution) is defined as follows.
	The input is a random matrix $\mat{A}\in\Z_q^{n\times m}$.
	The goal is to output a vector $\vec{x}\in\Z^m$ such that $\mat{A}\vec{x}=0\bmod q$ and $0<\|\vec{x}\|_p \leqslant \ell$.
\end{definition}

The $\ell_2$-norm case was first introduced by Ajtai~\cite{Ajtai96} and admits worst-case to average-case reduction for certain parameter regimes. It has found numerous applications in cryptography, including the construction of collision-resistant hash functions, digital signatures, identification schemes, and more recently, zero-knowledge proofs. The hardness of the recently NIST standardized signature scheme Dilithium~\cite{Ducas2017CRYSTALSDilithiumAS} relies on the hardness of the (module version of the) SIS problem under the $\ell_{\infty}$-norm.



\subsection{Quantum Search}\label{sec:qaa}

A fundamental building block of many algorithms (including several of this paper) is quantum amplitude amplification (QAA).

\begin{theorem}[\cite{brassard2002quantum}, adapted]
Let $\mathcal{A}$ be a quantum algorithm with no measurements constructing a quantum state:
\[ \mathcal{A} \ket{0} = \alpha \ket{\psi} \ket{1} + \sqrt{1-\alpha^2} \ket{\phi} \ket{0} \]
where $\alpha$ is a real number and $\ket{\psi}, \ket{\phi}$ are normalized quantum states. Let $QAA(\mathcal{A}) = - \mathcal{A} O_0 \mathcal{A}^{\dagger} O$ be the \emph{QAA iterate}, where $O$ flips the phase of $\ket{\psi}$ (basis states marked by 1) only, and $O_0$ flips the phase of the all-0 basis state only. Then, for all $k$:
\[ QAA(\mathcal{A})^k \mathcal{A} \ket{0} = \sin( (2k+1) \arcsin \alpha) \ket{\psi} \ket{1} + \cos((2k+1) \arcsin \alpha) \ket{\phi} \ket{0} \enspace. \]
\end{theorem}

This theorem tells us that if the algorithm $\mathcal{A}$ ``succeeds'' (output marked by 1) with probability $\alpha^2 = p$, then by making $\bigO{1/\alpha} = \bigO{1/\sqrt{p}}$ iterations, we will measure a success with constant probability. This can be increased arbitrarily close to 1 if $p$ is exactly known when running the algorithm~\cite{brassard2002quantum}.

In our paper, we will also need to succeed when only an upper bound on the probability $p$ is given, but it is not exactly known. In that case one can rely on fixed-point quantum search~\cite{yoder2014fixed}. It modifies the iteration operators to ensure a probability of success $1-\delta^2$ with $\bigO{ \log(1/\delta) \frac{1}{\sqrt{p}}}$ iterations, where $p$ is a lower bound on the probability of success of $\mathcal{A}$. Therefore, we can also have a probability of success exponentially close to 1 at the cost of a polynomial factor in the complexity.


\subsection{Quantum Algorithm for Finding the Maximum}\label{sec:Finding the maximum in a register}

Within the quantum attack on LWE of~\cite{DBLP:conf/eurocrypt/PoulyS24}, one has to be able, at some point, to find and measure the \emph{maximum} in a superposition, in the following sense. We have a function: $s~: \{0,1\}^\ell \to \mathbb{R}$ and a unitary operator computing $s$ in superposition: $\ket{i} \ket{0} \mapsto \ket{i} \ket{s(i)}$. The goal is to find $i = \mathrm{argmax}_{ \{0,1\}^\ell } s$.
An algorithm that allows to do so is presented in~\cite{DBLP:journals/corr/quant-ph-9607014}. However, in the setting in which we will have to use this algorithm, we will have to deal with a function $s$ that is probabilistic. More precisely, we will use a function $s$ that, if called several times upon the same input $i$, will most of the time return the ``right'' value $s(i)$, and will sometime return a ``wrong'' value that differs from $s(i)$. 
To do so, we rely on a combination of maximum-finding with bounded-error quantum search~\cite{DBLP:conf/icalp/HoyerMW03}, proving the following result in Appendix~\ref{appendix:Finding the maximum in a register}. This is essentially a generalization of Theorem 8 in~\cite{DBLP:conf/eurocrypt/PoulyS24}.

\begin{restatable}{theorem}{thmmaxfinding}[{\cite{DBLP:journals/corr/quant-ph-9607014,DBLP:conf/icalp/HoyerMW03}}]\label{theorem: cost of finding the maximum with faulty function}
Let $s$ be an efficient randomized algorithm that, on input $i \in \{0,1\}^\ell$, returns $s(i)$, where $s(i)$ equals a value $s_i$ (``good'') with probability at least 9/10, and $s$ returns a ``faulty'' output with probability at most 1/10.~\autoref{alg:finding the maximum in a register}, parametrised by $k > 0$, finds the $i$ such that $s_i$ is maximal, with probability at least $1-1/2^k$. This algorithms calls $s$ a number of times in $O(k\sqrt{2^\ell})$, and uses $O(k\ell \sqrt{2^\ell})$ additional gates, and $\poly(\ell)$ qubits.
\end{restatable}

\subsection{Quantum Algorithm for Computing the Mean}

Let $\mathcal{C}$ be a quantum algorithm without measurements constructing a superposition over $\ell$ qubits: $\mathcal{C} \ket{0}^{\otimes \ell} = \ket{\psi}$. Let $\phi:\{0,1\}^\ell \to \R$ be a function. 
Let $Y$ be the random variable obtained by measuring the output of $\mathcal C \ket 0^{\otimes \ell}$, and applying $\phi$ on it. We use a result of~\cite{Montanaro} to estimate the mean value of $Y$, assuming that $\mathrm{Var}(Y) \leq v^2$ for some value $v$. 

\begin{theorem}[{\cite[Theorem 5 and Algorithm 3, adapted]{Montanaro}}]\label{thm:montanaro algo}
    For any $\epsilon > 0$, there exists a quantum algorithm that calls $\mathcal C$ (and its inverse) a total of $O((v/\epsilon)\log^{3/2}(v/\epsilon)\log \log(v/\epsilon))$ times, a quantum implementation of $\phi$ a number $O(\log(v/\epsilon)\log \log(v/\epsilon))$ of times, and that estimates the mean value of $Y$ up to additive error $\epsilon$, with success probability at least $2/3$.
\end{theorem}

\begin{remark}\label{remark:montanaro qubit count}
The algorithm of ~\autoref{thm:montanaro algo} heavily relies on the quantum amplitude estimation algorithm of~\cite{brassard2002quantum}. The number of qubits is $O(\log(1/\epsilon))$ summed with the number of qubits used by $\phi$ and $\mathcal{C}$.
\end{remark}


Given a full rank lattice $\L\in \R^n$, an element $\vec x\in \R^n$, and a real number $\sigma>0$, we apply this theorem to estimate $f_{\L,\sigma}(\vec x)$. Indeed, we saw in~\autoref{sec:gaussian lattice distrib} that, by ~\autoref{lemma:fourier serie equal to f} and~\ref{lemma:Fourier serie of f L,sigma}, for all $\vec x\in \R^n$,
\begin{equation}\label{eq:mean}
    f_{\L,\sigma}(\vec x) = \sum_{w\in \hat \L} D_{\hat \L,1/\sigma}(\vec w)\exp{2\pi i \langle \vec x,\vec w\rangle} = \sum_{w\in \hat \L} D_{\hat \L,1/\sigma}(\vec w)\cos(2\pi \langle \vec x,\vec w\rangle)\enspace.
\end{equation}
(the second equality is obtained by regrouping $\exp{2\pi i \langle \vec x,\vec w\rangle}$ and $\exp{2\pi i \langle \vec x,\vec -w\rangle}$). We can then bound the variance $\mathcal V$ of $\cos(2\pi \langle \vec x,\vec w\rangle)$ when $\vec w$ is sampled from $D_{\hat \L,1/\sigma}$:
\begin{align*}
    \mathcal V &= - \bigg(\sum_{w\in \hat \L} D_{\hat \L,1/\sigma}(\vec w)\cos(2\pi \langle \vec x,\vec w\rangle)\bigg)^2+\sum_{w\in \hat \L} D_{\hat \L,1/\sigma}(\vec w)\cos(2\pi \langle \vec x,\vec w\rangle)^2\\
    &\leqslant \sum_{w\in \hat \L} D_{\hat \L,1/\sigma}(\vec w)\cos(2\pi \langle \vec x,\vec w\rangle)^2 \leqslant \sum_{w\in \hat \L} D_{\hat \L,1/\sigma}(\vec w) = 1 \enspace.
\end{align*}
Therefore, to estimate $f_{\L,\sigma}(\vec x)$, we can use~\autoref{thm:montanaro algo} with $v=1$. We also replace $\hat \L$ by a finite subset, ensuring a good approximation by a total variation distance bound. Finally, we use the powering lemma~\cite{10.5555/11534.11537} to increase the success probability of the algorithm, so that it can be guaranteed to work for all inputs $\vec{x}$ simultaneously.


\begin{corollary}\label{corollary:montanaro bis}
    Consider a full rank lattice $\L\in \R^n$, a finite set $X \subseteq \R^n$, and real numbers $\epsilon,\sigma>0$. Let $\Omega$ be a finite subset of $\hat \L$ such that $d_{TV}(D_{\hat \L,1/\sigma},D_{\Omega,1/\sigma})\leqslant \gamma$ for some $\gamma>0$.
    Let $\mathcal{C}$ be a quantum algorithm without measurements that does: $\mathcal C : \ket 0^{\otimes \ell} \mapsto \ket{D_{\Omega,1/\sigma}}$ and $\phi : \vec w \mapsto \cos(2\pi \langle \vec x,\vec w\rangle)$.
    
    There exists a quantum algorithm that calls $\mathcal C$ a number
     $\bigOt{(1/\epsilon) \log(1/\delta)}$
     times, calls a quantum implementation of $\phi$ a number 
     $\bigOt{\log(1/\epsilon) \log (1/\delta)}$ 
     of times, and that for all $\vec{x} \in X$, estimates the value of $f_{\L,\sigma}(\vec x)$ up to additive error $\epsilon+2\gamma$, with success probability at least $1 - |X|\delta$. If $\mathcal{C}$ and $\phi$ use $\poly(\ell, \log(1/\epsilon))$ qubits, then this algorithm uses $\poly(\ell, \log(1/\epsilon))$ qubits.
\end{corollary}

\begin{proof}
Let us first assume that $\vec{x} \in X$ is fixed. By~\autoref{thm:montanaro algo}, the mean value that the algorithm outputs is an estimation $Y'$ of the value 
    \begin{equation*}
        \mu(\vec x) := \sum_{w\in \Omega} D_{\Omega,1/\sigma}(\vec w)\cos(2\pi \langle \vec x,\vec w\rangle)\enspace,
    \end{equation*}
    such that with probability at least $2/3$, we have $ \left|Y' -\mu(\vec x)\right|\leqslant \epsilon  \enspace.$

Next, we use~\cite[Powering Lemma 6.1]{10.5555/11534.11537} to increase the accuracy of the estimate. Given an error parameter $\delta > 0$, we run the estimation algorithm $O(\log(1/\delta))$ times. Each run is a sample of a random variable $Y'$ such that $|Y' - \mu(\vec x)| \leqslant \epsilon$ with probability 2/3. By taking the median $\widetilde \mu(\vec x)$ of all samples, we have $ \left|\widetilde \mu(\vec x) -\mu(\vec x)\right|\leqslant \epsilon$ with probability at least $1 - \delta$.
The probability that this holds for all $\vec x\in X$ simultaneously is at least $(1-\delta)^{|X|}\geqslant 1-|X|\delta$.

Then, by~\autoref{eq:mean}, we have 
    \begin{align*}
        f_{\L,\sigma}(\vec x) &= \sum_{w\in \hat \L} D_{\hat \L,1/\sigma}(\vec w)\cos(2\pi \langle \vec x,\vec w\rangle)\\
      \implies  f_{\L,\sigma}(\vec x) - \mu(\vec x) &= \sum_{w\in \hat \L} (D_{\hat \L,1/\sigma}(\vec w)-D_{\Omega,1/\sigma}(\vec w))\cos(2\pi \langle \vec x,\vec w\rangle)\\
        \implies \left| f_{\L,\sigma}(\vec x) - \mu(\vec x) \right| & \leq \sum_{w\in \hat \L} \left|D_{\hat \L,1/\sigma}(\vec w)-D_{\Omega,1/\sigma}(\vec w)\right| \leqslant 2\gamma \enspace.
    \end{align*}
    By the triangle inequality, we have $\left|f_{\L,\sigma}(\vec x)-\widetilde \mu(\vec x)\right|\leqslant \epsilon+2\gamma$ with probability at least $1-|X|\delta$, for all $\vec{x}$. From~\autoref{thm:montanaro algo} and~\autoref{remark:montanaro qubit count}, we get the complexities for a single call to the estimation algorithm, which we multiply by $\log (1/\delta)$ (the number of calls). This concludes the proof.
    \qed
\end{proof}

We note that in the algorithm of~\autoref{corollary:montanaro bis}, which contains multiple independent runs of Montanaro's algorithm, all measurements can be deferred to the end. Doing this only increases the qubit count by a multiplicative factor $\log (1/\delta)$ (since all runs must be performed coherently).

\section{Quantum Gaussian Sampling}\label{sec: quantum ways of samplig from pi}

In this section we explain how to sample elements from the Gaussian lattice distribution quantumly, by combining a quantum Klein sampler with quantum rejection sampling. Pouly and Shen~\cite[Section 5.2]{DBLP:conf/eurocrypt/PoulyS24} also noticed that the Klein sampler could be turned into a quantum algorithm, which they could use after a BKZ preprocessing of the lattice basis -- though they did not give any details on the quantum implementation. de Boer and Felderhoff~\cite{DBLP:journals/iacr/BoerF25} described a detailed quantum version of the Klein sampler (Algorithm 6.1 in~\cite{DBLP:journals/iacr/BoerF25}), which we reproduce in~\autoref{sec:gaussian-state} for completeness, with slightly different, somewhat simplified complexity estimates for our case. Afterwards, the remainder of this section is dedicated to the rejection sampling step.



Throughout this section, we work with approximate distribution states:
$$ \ket{D \pm 2^{-\nu}} := \sum_{x \in X} \alpha_x \Ket{x}, \ \alpha_x\in \left[\sqrt{D(x)} - 2^{-\nu},\sqrt{D(x)} + 2^{-\nu}\right]$$
where $\ket{D} = \sum_{x \in X} \sqrt{D(x)} \Ket{x}$ would have been our wanted state. The total variation distance between the two probability distributions can be bounded as:
$$ \frac{1}{2} \sum_{x \in X} \left| (\sqrt{D(x)} \pm 2^{-\nu})^2 - D(x) \right| \leq \frac{3}{2} |X| 2^{-\nu} \enspace. $$
Thus, it suffices to choose $\nu \gg \log_2 |X|$ to obtain a negligible total variation distance with respect to the target distribution.

\subsection{Preliminary: Estimating a Gaussian Mass}

An important operation, both for the quantum Klein sampler and quantum rejection sampling, is to estimate a Gaussian mass over a finite subset of $\Z$, i.e., for $\widetilde x\in \R$:
\begin{equation*}
    \rho_{\sigma,\widetilde{x}}(\iintv{-2^M,2^M}) = \sum_{k=-2^M}^{2^M} \exp{-\frac{\pi}{\sigma^2}(k-\widetilde{x})^2} \enspace.
\end{equation*}

For this estimation we can follow Lemma A.30 in~\cite{boer2022random}, which gives the explicit complexity of a method proposed by Kitaev and Webb~\cite{kitaev2008wavefunction}. In a nutshell, one starts from an approximate quantum implementation of $x \mapsto \exp{x}$ using fixed-point arithmetic, a summation over a smaller interval, and uses the Poisson formula to switch between the cases of a small or a large $\sigma$. The only technicality in our result is that we switch from a relative precision to an absolute one. In the rest of the paper, $M$ will always be large enough to satisfy the condition of~\autoref{corollary:gaussian-mass}.

\begin{corollary}\label{corollary:gaussian-mass}
There exists a universal constant $c$ such that, for any $\nu, \sigma, M$ with $\nu \geq \log \sigma$ and $\log \sigma + M\log(2) + \nu \leq c 2^{2M}$, the following holds.

There exists a quantum algorithm that has a gate count $\bigOt{\nu^{3/2}}$ and a qubit count $\bigOt{\nu}$, and that computes:
$$ \ket{ \widetilde{x} } \ket{0} \mapsto \ket{\widetilde{x}} \ket{\rho_{\sigma,\widetilde{x}}(\iintv{-2^M,2^M}) \pm 2^{-\nu}}\enspace . $$
\end{corollary}

\begin{proof}
We start by computing an estimation $\widetilde \rho$ of $\rho_{\sigma,\widetilde x}(\Z) = \sum_{k=-\infty}^{\infty} \exp{-\frac{\pi}{\sigma^2}(k-\widetilde{x})^2}$ with relative precision $2^{-\nu'}$, for some $\nu'$, using Lemma A.30 in~\cite{boer2022random}. The algorithm of~\cite{boer2022random} is unitary, has gate count $\bigOt{(\nu')^{3/2}}$ and qubit count $\bigOt{\nu'}$.
We have:
\begin{align*}
| \widetilde \rho - \rho_{\sigma,\widetilde{x}}(\iintv{-2^M,2^M}) | &\leq  | \widetilde \rho - \rho_{\sigma,\widetilde x} | + | \rho_{\sigma,\widetilde x} \ - \rho_{\sigma,\widetilde{x}}(\iintv{-2^M,2^M}) | \\
| \widetilde \rho - \rho_{\sigma,\widetilde{x}}(\iintv{-2^M,2^M}) | &\leq 2^{-\nu'} \rho_{\sigma,\widetilde x} + | \rho_{\sigma,\widetilde x} \ - \rho_{\sigma,\widetilde{x}}(\iintv{-2^M,2^M}) | \text{ By~\cite{boer2022random}} \\
& \leq 2^{-\nu'} \rho_{\sigma,\widetilde x} + O(\sigma\times 2^M \exp{-\pi 2^{2M}}) \text{ By~\autoref{lemma: inequality on gaussian masses on finite sets}} \\
& \leq 2^{-\nu'} (3\sigma + 4) + O(\sigma  2^M \exp{-\pi 2^{2M}})  \text{ By Lemma~\ref{lemma:rho sigma c} and~\ref{lemma:rho Z approx}. }
\end{align*}

We take $\nu' = \bigO{\nu}$ and bound the second term thanks to the condition on $M$, concluding the proof. 
\qed
\end{proof}

\subsection{Step 1: Constructing a Gaussian State}\label{sec:gaussian-state}

In the following, we give the quantum Klein sampler algorithm from~\cite{DBLP:journals/iacr/BoerF25}. 

\paragraph{One-dimensional Gaussian State.}
The algorithm starts by approximating a one-dimensional quantum Gaussian state, using the method of Kitaev and Webb~\cite{kitaev2008wavefunction}, whose complexity is detailed by de Boer~\cite{boer2022random}.
We adapt Theorem A.29 from de Boer~\cite{boer2022random}, who constructs a one-dimensional centered Gaussian state. In our case we just need to shift the Gaussian by a variable quantity $\widetilde{x}$, which is also made possible by Algorithm 10 in~\cite{boer2022random}.

The input $\widetilde x$, which is real, is stored in qubits using a fixed-point approximation with a sufficient precision. This is taken into account by the $\bigO{M + \nu}$ memory complexity. The idea of the algorithm is to produce the Gaussian state bit by bit, recursively. At each bit, one computes a Gaussian sum to determine the left and right weight, then produces a superposition of $\ket{0} / \ket{1}$ with appropriate amplitudes, and uncomputes the sum. The entire algorithm runs without measurements.

\begin{theorem}[A.29 in~\cite{boer2022random}, adapted to our notation]\label{thm:boer}
For $M \in \mathbb{N}$, $\nu \in \mathbb{N}$ and $\sigma $, there exists a quantum algorithm that computes:
\[ \ket{\widetilde{x}} \ket{0} \mapsto \ket{\widetilde{x}} \ket{ D_{\sigma,\widetilde x}(\iintv{-2^M + \lfloor \widetilde x\rceil,2^M+\lfloor \widetilde x\rceil})} \]
within trace distance $\exp{- 2^{2M} / (2\sigma^2)} + M 2^{-\nu}$, using $\bigO{M \nu^{3/2} \polylog{\nu}}$ quantum gates and $\bigO{M + \nu}$ qubits.
\end{theorem}

Recall that the trace distance between two (pure) states $\ket{\varphi}$ and $\ket{\psi}$ is $\sqrt{1-|\braket{\varphi |\psi}|^2}$, and is higher than the total variation distance of the probability distributions resulting from a measurement of the states.

By adapting the precision used inside the algorithm, we can bound the trace distance exactly as $2^{-\nu}$ if the complexity formula is slightly changed. It suffices to choose $M$ large enough compared to $\sigma$ and $\nu$. This very mild constraint will easily be satisfied in our algorithms.




\begin{corollary}\label{cor:quantum-states}
There exists a universal constant $c$ such that, for any $\sigma $, for any $\nu$, for any $M \geq c (\log \sigma + \log \nu)$, 
there exists a quantum algorithm that uses using $\bigO{M + \nu}$ qubits and $\bigOt{ M \nu^{3/2}}$ quantum gates, and that computes:
\[ \ket{\widetilde{x}} \ket{0} \mapsto \ket{\widetilde{x}} \ket{ D_{\sigma,\widetilde x}(\iintv{-2^M + \lfloor \widetilde x\rceil,2^M+\lfloor \widetilde x\rceil}) \pm 2^{-\nu} }\enspace . \]
\end{corollary}

\begin{proof}
We run the algorithm of~\autoref{thm:boer} using a slightly larger precision $\nu' = \bigO{\nu + \log M}$. Using $M \geq c  \log \sigma$ we can ensure that $\exp{- 2^{2M} / (2\sigma^2)}$ is dominated by $\exp{- 2^{2M-1}}$. Using $\nu'$ we also ensure that the second term $M 2^{-\nu'}$ in~\autoref{thm:boer} is dominated by $2^{-\nu}$. The distance becomes: $\exp{- 2^{2M-1}} + 2^{-\nu}$. Finally since $M \geq c \log \nu$ the first term can be bounded by $2^{-\nu}$ as well. The asymptotic gate and qubit counts are unchanged, except for the dependency in $\log M$, which is absorbed by the $\widetilde{\mathcal{O}}$ notation.\qed
\end{proof}

\paragraph{Many-dimensional Gaussian State.}
We now detail the algorithm of~\cite{DBLP:journals/iacr/BoerF25}. 
Recall that the Klein distribution $q$ parametrised by $\sigma$, $B\in GL_m(\R)$, $\vec c\in \R^m$ and $M\in \N$, $M>1$, is the distribution of lattice vectors obtained by running~\autoref{alg:Klein's algorithm} classically. In the quantum setting,~\autoref{alg:quantum-klein} allows to construct an approximation of it, as a quantum state: $\Ket{q \pm 2^{-\nu}}$. The idea is to run the same procedure as classically, where one-dimensional Gaussian samples are replaced by the construction of Gaussian quantum states. Note that we assume that the QR decomposition of $B$ and $\vec c' = Q^{-1}\vec c$ are precomputed classically.

\begin{algorithm}[tb]
    \caption{Quantum Klein sampling (variant of Algorithm 6.1 in~\cite{DBLP:journals/iacr/BoerF25}).}\label{alg:quantum-klein}
\begin{algorithmic}[1]
\Statex \textbf{Input:} $B = QR$, $\sigma$, $\vec c' = Q^{-1}\vec c$, $M\in \N$, $M>1$
\smallskip
\Statex \textbf{Output:} $\Ket{q}$
\medskip
\State Initialize a quantum state $\Ket{\vec{0}} = \Ket{0} \Ket{0} \ldots \Ket{0}$ with $m$ registers
    \For{$i=m \to 1$}
        \State $\sigma_i\leftarrow \sigma/|R_{i,i}|$
        \State Let the current state be $\sum \Ket{0} \Ket{0} \ldots \Ket{0} \alpha_{x_{i+1}, \ldots, x_m} \Ket{x_{i+1}} \ldots \Ket{x_m}$
        \State Compute:
        $$ \sum \Ket{0} \Ket{0} \ldots \Ket{0} \alpha_{x_{i+1}, \ldots, x_m} \Ket{x_{i+1}} \ldots \Ket{x_m} \Ket{\widetilde x_i} $$
        where $\widetilde x_i = \frac{c'_i-\sum_{j=i+1}^mR_{i,j}x_j}{R_{i,i}}$
        \State Construct a Gaussian state in the register $i$:
        $$\sum \Ket{0} \Ket{0} \ldots \Ket{D_{\sigma_i, \widetilde{x}_{i}}(\iintv{-2^M+\lfloor \widetilde x_i \rceil,2^M+\lfloor \widetilde x_i \rceil})} \alpha_{x_{i+1}, \ldots, x_m} \Ket{x_{i+1}} \ldots \Ket{x_m}$$
	\EndFor
    \State  Multiply by $B$ in-place on the quantum state
\end{algorithmic}
\end{algorithm}

Note that in~\autoref{alg:quantum-klein}, we can compress the space by storing $(x_i - \floor{\widetilde{x}_i})$ in Step 5, which is within $\iintv{-2^M,2^M}$. Then at each loop iterate, we can recompute the center of the previous distribution. We choose to simplify the notations in order to make the algorithm easier to understand. 


\begin{theorem}\label{theorem: Klein distrib}
There exists a universal constant $c$ such that the following holds.

Let $q$ be the Klein distribution resulting from~\autoref{alg:Klein's algorithm}, which is parametrised by $\sigma$,
$B\in GL_m(\R)$, $\vec c$ and $M\in \N$, $M>1$. Let $r = \max \mathsf{size}(R_{ij}, c_i', \sigma)$. Let $\nu \geq c (r + \log m)$ and $M \geq c(\log r + \log \nu + m)$.
There exists a quantum circuit of qubit count $\bigOt{m (M + mr + \nu ) }$ and gate count $ \bigOt{ m^2(mr + M + \nu) + m M \nu^{3/2}}$ that prepares $\ket{q \pm 2^{-\nu}}$.
\end{theorem}

Here $\mathsf{size}$ denotes the bit-size required to represent the constants $(R_{ij}, c_i', \sigma)$ as rational numbers.

\begin{proof}
The quantum circuit this theorem refers to is~\autoref{alg:quantum-klein}. We bound precisely its cost and precision.

Within the algorithm, while $x_i$ are always integers, each $\widetilde{x_i}$ is computed using a fixed-point approximation to precision $\nu' = \bigO{\nu}$. Since each $R_{i,j}$ and $c_i'$ are smaller than $2^r$, at step $i$, the numerator of $\widetilde{x_i}$ can be bounded by: $|c_i'| + 2^r (\max_{j \leq i} |x_j|)(m-i) \leq 2^r + 2^r (m-1) \max_{j \leq i} |x_j|$, while the denominator can be bounded by $2^r$. Since $|x_m| \leq 2^M$, we can unroll this recurrence relation and obtain: $\forall i,  2^r |\widetilde{x_i}| \leq ( 2^r m)^i 2^M + \bigO{ ( 2^r m)^{i-1}  2^r } = \bigO{ (2^r m)^i 2^M }$. As a consequence, each coordinate $x_i$ requires at most $M + m \log ( 2^r m) + \nu = \bigOt{M + m r + \nu}$ bits of space to be represented. This means that multiplications required to compute $\widetilde{x_i}$ can be implemented within $\bigOt{M + m r + \nu}$ gates, and there are $\bigO{m^2}$ of them. All operations in Step 4 therefore cost $\bigOt{ m^2(m r + M + \nu) }$ gates.

Finally, the result of~\autoref{cor:quantum-states} is used $m$ times, with the different $\sigma_i$ (using fixed-point approximations again). Each step runs without measurements. The cost of this part is $\bigOt{m M \nu^{3/2} }$, and the qubit count $\bigO{\nu + M}$ is dominated by the rest of the algorithm.

At the end of~\autoref{alg:quantum-klein}, the amplitudes in the quantum state are the product of $m$ approximations of Gaussian amplitudes. The final error is approximately the sum of $m$ errors. In order to have a final precision $\nu$, we use a precision $\nu' = \nu + \log_2 m$ internally in the Gaussian state creation and arithmetic. Since $\nu \geq c \log m$ we still have $\nu' = \bigO{\nu}$. Thus, we can reach a precision $2^{-\nu}$ with a gate count $ \bigOt{m^2(m r + M + \nu) + m M \nu^{3/2}}$ and a qubit count $\bigOt{m (M + m r + \nu ) }$.
\qed
\end{proof}

This result is similar to Lemma L.5 in~\cite{DBLP:journals/iacr/BoerF25}, as the two algorithms are equivalent (only the order of operations in~\cite{DBLP:journals/iacr/BoerF25} is performed in a slightly different way, as the quantum state is rescaled between each iteration of the for loop).

\subsection{Step 2: Transforming the Distribution}\label{sec:gaussian-distrib}

Let $\Omega$ be a finite probability space within $\L(B)$. Given a quantum circuit that produces quantum samples for a distribution $D_1$ over $\Omega$:
$$ \Ket{0} \xmapsto{\mathcal{A}_1} \Ket{D_1} := \sum_{\vec{x} \in \Omega} \sqrt{D_1(\vec{x})} \Ket{\vec{x}} $$
we use \emph{quantum rejection sampling}~\cite{DBLP:journals/toct/OzolsRR13} to construct quantum samples for another distribution $D_2$ over $\Omega$, up to some approximation. This section starts with a generic statement (\autoref{theorem: amplitude transduction}) which is essentially from~\cite{DBLP:journals/toct/OzolsRR13}, and continues with the details in our case.

\subsubsection*{Quantum Rejection Sampling.}
The main building block from quantum rejection sampling is \emph{amplitude transduction}. We use the method of~\cite{sanders2019black} to implement the following unitary with $\bigO{\nu}$ basic gates:
\begin{equation}
\Ket{\alpha \pm 2^{-\nu}} \Ket{0} \mapsto \Ket{\alpha \pm 2^{-\nu}}\left( (\alpha \pm 2^{-\nu}) \Ket{0} + \Ket{*} \right)
\end{equation}
where $\Ket{*}$ is a superposition of non-zero basis vectors, and $\alpha \pm 2^{-\nu}$ is a fixed-point encoding, at precision of $\nu$ bits, of an amplitude between 0 and 1.
We then have the following result.

\begin{theorem}\label{theorem: amplitude transduction}
Let $\mathcal{A}_1$ be a quantum circuit (without measurement) that produces an approximation of $\Ket{D_1}$: $\mathcal{A}_1\Ket{0} = \Ket{D_1 \pm 2^{-\nu}}$. Let $w \geq 0$ be such that for all $\vec x\in \Omega$, $D_2(\vec{x})/ D_1(\vec{x}) \leq w$. Let $\mathcal{B}$ be a quantum circuit that computes a fixed-point approximation of $  \sqrt{ \frac{1}{w} \frac{D_2(\vec{x})}{ D_1(\vec{x}) } } $:
$$ \mathcal{B} \Ket{\vec{x}} \Ket{0} = \Ket{\vec{x}} \Ket{  \sqrt{\frac{1}{w} \frac{D_2(\vec{x})}{ D_1(\vec{x}) } } \pm 2^{-\nu} }  $$
then there exists a quantum algorithm that outputs $\Ket{D_2 \pm 2^{-\nu+1}}$ after making an average $\bigO{ \sqrt{w} }$ calls to $\mathcal{A}_1$ and $\mathcal{B}$.
\end{theorem}

\begin{proof}
From the output of $\mathcal{A}_1$, we apply $\mathcal{B}$, and obtain:
\begin{equation}
 \sum_{\vec{x} \in \Omega} (\sqrt{D_1(\vec{x})} \pm 2^{-\nu}) \Ket{\vec{x}} \Ket{ \frac{1}{w} \sqrt{ \frac{D_2(\vec{x})}{ D_1(\vec{x}) } } \pm 2^{-\nu} } \enspace.
\end{equation}
Since $ \sqrt{\frac{1}{w} \frac{D_2(\vec{x})}{ D_1(\vec{x}) } } \pm 2^{-\nu}$ is a $\nu$-bit fixed-point approximation of a number between 0 and 1, we apply the algorithm of~\cite{sanders2019black} to an ancilla register $\Ket{0}$:
\begin{multline*}
\sum_{\vec{x} \in \Omega} (\sqrt{D_1(\vec{x})} \pm 2^{-\nu} )\Ket{\vec{x}} \Ket{  \sqrt{\frac{1}{w} \frac{D_2(\vec{x})}{ D_1(\vec{x}) } } \pm 2^{-\nu} } \left( \left( \sqrt{\frac{1}{w}  \frac{D_2(\vec{x})}{ D_1(\vec{x}) } } \pm 2^{-\nu} \right) \Ket{0} + \Ket{*} \right) \\
= \left( \left(\frac{1}{\sqrt{w}} \sum_{\vec{x}} \sqrt{D_2(\vec{x})} \pm 2^{-\nu+1}\right) \Ket{\vec{x}} \Ket{  \sqrt{ \frac{1}{w}\frac{D_2(\vec{x})}{ D_1(\vec{x}) } } \pm 2^{-\nu} }  \right) \Ket{0} + \Ket{*}
\end{multline*}
We uncompute $\mathcal{B}$.
Then, we use quantum amplitude amplification to project on the component $\Ket{0}$. The total amplitude over this component being $1/w$, using quantum amplitude amplification we can succeed with $O(w)$ iterations.
\qed
\end{proof}




\subsubsection{Our Application of Rejection Sampling.}

We use quantum rejection sampling with the following setup. The initial distribution $D_1$ is the Klein distribution $q$. The target distribution $D_2$ is a Gaussian lattice distribution $D_{\Omega, \sigma}$. The domain $\Omega$ is the probability space of the Klein Algorithm (\autoref{alg:Klein's algorithm}) parametrised by $\sigma$, $\vec c=0$, and the integer $M>1$.

We need to determine a bound $w$ on the ratio of the two distributions: $w := \max_{\vec{x} \in \Omega} \frac{D_2(\vec{x})}{D_1(\vec{x})}\enspace,$
and we need to be able to compute $\sqrt{\frac{1}{w}  \frac{ D_2(\vec{x}) }{D_1(\vec{x}) }}$ on input $\vec{x}$ (though it is not necessary to know $w$ itself to run the algorithm). To achieve this, recall that the Gaussian lattice distribution that we are targeting has expression:
\begin{equation*}
\forall \vec{x} \in \Omega,  D_{\Omega, \sigma}(\vec x) = \frac{\rho_{\sigma}(\vec x)}{\rho_{\sigma}(\Omega)} = \frac{\exp{-\frac{\pi}{\sigma^2}||\vec x||^2}}{\sum_{\vec y\in \Omega} \exp{-\frac{\pi}{\sigma^2}||\vec y||^2}}
\end{equation*}
Following~\autoref{equation:q expression}, the Klein distribution has expression:
\begin{equation*}
\forall \vec{x} \in B^{-1}\Omega, q(B\vec{x}) = \frac{\rho_{\sigma}(B\vec x)}{\prod_{i=1}^m \rho_{\sigma_{i},\widetilde{x}_{i}}(\iintv{-2^M+\lfloor \widetilde x_i\rceil, 2^M+\lfloor \widetilde x_i\rceil})}
\end{equation*}
Given a vector $\vec x\in \R^n$, denote by $\widetilde x_{B,i}$ the coefficient $\widetilde y_i = \frac{-\sum_{j=i+1}^m R_{i,j}y_j}{R_{i,i}}$, where $\vec y = B^{-1}\vec x$. We rewrite the expression of $q$ as 
\begin{gather*}
\forall \vec{x} \in \Omega, q(\vec{x}) = \frac{\rho_{\sigma}(\vec x)}{\prod_{i=1}^m \rho_{\sigma_{i},\widetilde x_{B,i}}(\iintv{-2^M, 2^M}+\lfloor \widetilde x_{B,i}\rceil)}\enspace .\\
\Rightarrow \frac{D_{\Omega, \sigma}(\vec x)}{q(\vec x)}= \frac{\prod_{i=1}^m \rho_{\sigma_{i},\widetilde x_{B,i}}(\iintv{-2^M, 2^M}+\lfloor \widetilde x_{B,i}\rceil)}{ \rho_{\sigma}(\Omega) } \enspace.
\end{gather*}

According to~\autoref{lemma: inequality on gaussian masses on finite sets}, we always have 
\begin{align*}
    \rho_{\sigma_{i},\widetilde x_{B,i}}(\iintv{-2^M, 2^M}+\lfloor \widetilde x_{B,i}\rceil)&\leqslant C\rho_{\sigma_{i},\lfloor \widetilde x_{B,i}\rceil }(\iintv{-2^M, 2^M}+\lfloor \widetilde x_{B,i}\rceil)\\
    &\leqslant C\rho_{\sigma_{i}}(\iintv{-2^M, 2^M}),
\end{align*}
where $C:=(1-2^M\sqrt{2\pi e}\times \exp{-\pi 2^{2M}})^{-1}$. So, for all $\vec x\in \Omega$, we have 
\begin{equation*}\label{eq:w}
    \frac{D_{\Omega, \sigma}(\vec x)}{q(\vec x)} \leqslant \frac{ C^m\prod_{i=1}^{m} \rho_{\sigma_{i}}(\iintv{-2^M,2^M}) }{ \rho_{\sigma}(\Omega) } := w \enspace.
\end{equation*}
This gives our definition of $w$. 
Note that the constant $C$ is easy to compute, and gets closer to 1 as $M$ increases. 
Computing $w$ itself would be difficult, but it is not necessary. For all $\vec{x} \in \Omega$, we have:
\begin{align*}
    \frac{1}{w} \times \frac{ D_{\Omega, \sigma}(\vec x) }{q(\vec x)}  & = \frac{\rho_{\sigma}(\Omega) }{C^m\prod_{i=1}^{m} \rho_{\sigma_{i}}(\iintv{-2^M,2^M})}  \times \frac{\prod_{i=1}^m \rho_{\sigma_{i},\widetilde x_{B,i}}(\iintv{-2^M, 2^M}+\lfloor \widetilde x_{B,i}\rceil)}{ \rho_{\sigma}(\Omega) } \\
    &= \frac{\prod_{i=1}^m \rho_{\sigma_{i},\widetilde x_{B,i}}(\iintv{-2^M, 2^M}+\lfloor \widetilde x_{B,i}\rceil)}{C^m\prod_{i=1}^{m} \rho_{\sigma_{i}}(\iintv{-2^M,2^M})} \enspace.
\end{align*}

Let us then explain how to compute this quantity quantumly.
\begin{lemma}\label{lemma:implem of B}
Let $r = \max \mathsf{size}(R_{ij}, c_i', \sigma)$ as in~\autoref{theorem: Klein distrib}. Assume that the constant $C$ is precomputed classically and exactly, and $\nu \geq r$. There exists an implementation of the quantum circuit $\mathcal B:\ket{\vec x}\ket{0} \to \ket{\vec x}\Ket{\sqrt{\frac{1}{w} \times \frac{ D_2(\vec{x}) }{D_1(\vec{x}) }}\pm 2^{-\nu}}$ that has a gate count $\bigOt{m^2(mr + M + \nu) + m \nu^2}$ and uses $\bigOt{ mr + M + \nu }$ qubits.
\end{lemma}

\begin{proof}
In order to compute $\frac{1}{w} \times \frac{ D_2(\vec{x}) }{D_1(\vec{x}) }$, we need to estimate $2m$ Gaussian masses, using~\autoref{corollary:gaussian-mass}. Then, we note that, since $\rho_{\sigma_{i}}(\iintv{-2^M,2^M}) \geqslant 1$, we have
\begin{align*}
    \frac{1}{\rho_{\sigma_{i}}(\iintv{-2^M,2^M})+2^{-\nu''}} &= \frac{1+2^{-\nu''}/\rho_{\sigma_{i}}(\iintv{-2^M,2^M})-2^{-\nu''}/\rho_{\sigma_{i}}(\iintv{-2^M,2^M})}{\rho_{\sigma_{i}}(\iintv{-2^M,2^M})+2^{-\nu''}}\\
    &= \frac{1}{\rho_{\sigma_{i}}(\iintv{-2^M,2^M})} - \frac{2^{-\nu''}}{\rho_{\sigma_{i}}(\iintv{-2^M,2^M})(\rho_{\sigma_{i}}(\iintv{-2^M,2^M})+2^{-\nu''})}\\
    &\geqslant \frac{1}{\rho_{\sigma_{i}}(\iintv{-2^M,2^M})} - 2^{-\nu''}
\end{align*}
\begin{equation*}
\text{and likewise  } \qquad  \frac{1}{\rho_{\sigma_{i}}(\iintv{-2^M,2^M})-2^{-\nu''}}\leqslant \frac{1}{\rho_{\sigma_{i}}(\iintv{-2^M,2^M})} + 2^{-\nu''} \enspace.
\end{equation*}
Therefore, assuming $\rho_{\sigma_{i}}(\iintv{-2^M,2^M})$ and $\rho_{\sigma_{i},\widetilde x_{B,i}}(\iintv{-2^M, 2^M}+\lfloor \widetilde x_{B,i}\rceil)$ are computed with precision $2^{-\nu''}$, we can compute $\frac{\rho_{\sigma_{i},\widetilde x_{B,i}}(\iintv{-2^M, 2^M}+\lfloor \widetilde x_{B,i}\rceil)}{\rho_{\sigma_{i}}(\iintv{-2^M,2^M})}$ up to precision 
\begin{equation*}
    2^{-\nu''}\left(\rho_{\sigma_{i},\widetilde x_{B,i}}(\iintv{-2^M, 2^M}+\lfloor \widetilde x_{B,i}\rceil) + 1/\rho_{\sigma_{i}}(\iintv{-2^M,2^M})\right)+2^{-2\nu''}\enspace.
\end{equation*}
We already know that we have $1/\rho_{\sigma_{i}}(\iintv{-2^M,2^M})\leqslant 1$, and by Lemma~\ref{lemma:rho sigma c} and~\ref{lemma:rho Z approx}, we have:
\begin{equation*}
\rho_{\sigma_{i},\widetilde x_{B,i}}(\iintv{-2^M, 2^M}+\lfloor \widetilde x_{B,i}\rceil)\leqslant \rho_{\sigma_{i},\widetilde x_{B,i}}(\Z)\leqslant \rho_{\sigma_{i}}(\Z)\leqslant 3\sigma_i + 4 \enspace.
\end{equation*}
So we can compute $\frac{\rho_{\sigma_{i},\widetilde x_{B,i}}(\iintv{-2^M, 2^M}+\lfloor \widetilde x_{B,i}\rceil)}{\rho_{\sigma_{i}}(\iintv{-2^M,2^M})}$ up to precision $2^{-\nu''} (3\sigma_i + 4) +2^{-2\nu''}$.
Setting $\nu'' = \nu'+\bigO{ \log(\sigma_i + 1) }$ allows to get an additive precision of $2^{-\nu'}$. Finally, since $\sigma_i = \sigma / R_{i,i}$ and since we defined $r \geq \max \mathsf{size}( R_{i,i}, \sigma)$, we can bound $\nu'' = \nu' + \bigO{r}$. By~\autoref{lemma: inequality on gaussian masses on finite sets}, we have
\begin{equation}\label{eq:amplitude-bound-1}
    \frac{\rho_{\sigma_{i},\widetilde x_{B,i}}(\iintv{-2^M, 2^M}+\lfloor \widetilde x_{B,i}\rceil)}{C\rho_{\sigma_{i}}(\iintv{-2^M,2^M})}\leqslant 1\enspace.
\end{equation}
Therefore, if we can compute each $\frac{\rho_{\sigma_{i},\widetilde x_{B,i}}(\iintv{-2^M, 2^M}+\lfloor \widetilde x_{B,i}\rceil)}{\rho_{\sigma_{i}}(\iintv{-2^M,2^M})}$ up to precision $2^{-\nu'}$, then we can compute $\frac{1}{w} \times \frac{ D_2(\vec{x}) }{D_1(\vec{x}) }$ up to precision $m 2^{-\nu'}$. Setting $\nu' = \nu+\log(m)$ allows precision $2^{-\nu}$. This means that $\nu'' = \log m + \nu + \bigO{r}$ is the precision required for estimating the Gaussian masses. 

For each $i$, we then apply transduction separately:
\begin{equation}
\ket{\vec{x}} \ket{0} \mapsto \ket{\vec{x}}\left( \sqrt{  \frac{\rho_{\sigma_{i},\widetilde x_{B,i}}(\iintv{-2^M, 2^M}+\lfloor \widetilde x_{B,i}\rceil)}{C\rho_{\sigma_{i}}(\iintv{-2^M,2^M})}  }\ket{0} + \ket{*} \right) \enspace.
\end{equation}
The quantity $\frac{1}{w} \times \frac{ D_2(\vec{x}) }{D_1(\vec{x}) }$ then appears naturally as the amplitude over the 0-state. 
The gate and qubit count can be summarized as follows:
\begin{itemize}
\item We need to recompute the $\widetilde x_{B,i}$, similarly to~\autoref{alg:quantum-klein}. The gate count will be $\bigOt{m^2(mr + M + \nu)}$ gates. We also need $\bigOt{ mr + M + \nu }$ qubits to store the $x_i$, but this is the input of the algorithm.
\item We need to compute $2m$ times a Gaussian mass at precision $\log m + \nu + \bigO{r}$. By~\autoref{corollary:gaussian-mass}, this takes $\bigOt{m (\log m + \nu + r)^{3/2}} = \bigOt{m (\nu + r)^{3/2}}$ gates and $\bigOt{\log m + \nu + r}$ qubits.
\item We need to compute $m$ times a square root of a number smaller than 1, with additive precision $2^{-\nu}$.
\end{itemize}
For the latter, a simple method which uses $\bigO{\nu}$ space and $\bigO{\nu^2}$ gates is to use an iterative integer square root algorithm~\cite{knuth2011taocp2}.
Since we assumed that $\nu \geq r$, we can simplify the gate count as follows:
\begin{equation*}
\bigOt{m^2(mr + M + \nu) + m (\nu + r)^{3/2} + m \nu^2} = \bigOt{m^2(mr + M + \nu) + m \nu^2}
\end{equation*}
and the qubit count is $\bigOt{ mr + M + \nu }$ (dominated by the storage of $x_i$).
\qed
\end{proof}


\subsubsection{Constructing a New Gaussian State.}
We can now combine the quantum Klein sampler (\autoref{theorem: Klein distrib}), and the unitary of~\autoref{lemma:implem of B} as components of the quantum rejection sampling algorithm (\autoref{theorem: amplitude transduction}).

%

\begin{theorem}\label{theorem: performance of quantum gaussian sampling algorithm}
Consider a lattice $\L$, and $\Omega$ a finite subset of $\L$ that corresponds to the probability space of the output distribution of~\autoref{alg:Klein's algorithm}, parametrised by $\sigma$, $\vec c=0$ and $M\in \N$. Let $r = \max \mathsf{size}(R_{ij}, c_i', \sigma)$ as in~\autoref{theorem: Klein distrib} and assume $M \geq c(\log r + \log \nu)$ for some universal constant $c$.

There exists a quantum algorithm of gate count:
\begin{equation*}
\sqrt{\frac{ C^m\prod_{i=1}^{m} \rho_{\sigma_{i}}(\iintv{-2^M,2^M}) }{ \rho_{\sigma}(\Omega) }} \times \bigOt{ m M (\nu + mM + m^2 r)^2 }
\end{equation*}
and of qubit count $\bigOt{m (\nu + mM + m^2 r)} $ that outputs a distribution $D$ such that $d_{TV}(D,D_{\L,\sigma})\leqslant 2^{-\nu}$. 
\end{theorem}

\begin{proof}
According to~\autoref{theorem: amplitude transduction}, we need $\bigO{\sqrt{w}}$ calls to the algorithms of~\autoref{theorem: Klein distrib} and~\autoref{lemma:implem of B}, in order to produce the wanted distribution, where: $w =\frac{ C^m \prod_{i=1}^{m} \rho_{\sigma_{i}}(\iintv{-2^M,2^M}) }{ \rho_{\sigma}(\Omega) } $ (\autoref{eq:w}). 

We set the precision of the unitaries to $\nu' = \nu+\log_2(|\Omega|)$, where $\log_2 |\Omega| = m \log_2( (2^r m)^m 2^M) = \bigOt{ m M + m^2 r} $ due to the bounding of $x_i$ given previously. 

Due to this, the qubit count is: $\bigOt{ m(mr + M + \nu') } = \bigOt{m \nu'}$. The gate count of the quantum Klein sampler is $\bigOt{ m^2 \nu' + m M (\nu')^{3/2} } = \bigOt{ m M (\nu')^{3/2} }$ (observe that $\nu'$ is bigger than $m^2$). The gate count of the circuit $\mathcal{B}$ is given by $\bigOt{m^2 \nu' + m (\nu')^2} = \bigOt{ m (\nu')^2 }$. Summing both gives us the gate count for the iterate of the rejection sampling algorithm, which we upper bound by $\bigOt{ m M (\nu')^2 }$ for simplicity.

The precision chosen ensures that $|\Omega|2^{-(\nu+\log_2(|\Omega|))} = 2^{-\nu}$, so $d_{TV}(D,D_{\Omega,\sigma})$ is upper bounded by $2^{-\nu}$. Then, by~\autoref{eq: dtv of gaussian distrib over finite space vs infinite space} we have:
    \begin{align*}
        d_{TV}(D,D_{\L,\sigma})&\leqslant d_{TV}(D,D_{\Omega,\sigma}) + d_{TV}(D_{\Omega,\sigma},D_{\L,\sigma}) \leqslant 2^{-\nu}+\rho_\sigma(\L\backslash \Omega) \enspace. 
    \end{align*}
   
Finally, we use~\autoref{prop: bound on some dtv and gaussian mass} to bound $d_{TV}(D,D_{\L,\sigma})$ as follows:
\begin{equation}
d_{TV}(D,D_{\L,\sigma}) \leq 2 \left(2^{M-1}\sqrt{2\pi e}\right)^m\times \exp{-\pi\cdot m\cdot 2^{2M-2}}  (3\sigma / \min_i \| \widetilde{\vec b}_i\| + 4)^m
\end{equation}
Choosing $M$ a multiple of $\log r + \log \nu$ (at least) we can make this smaller than $2^{-\nu}$, finishing the proof. 
\qed
\end{proof}



\subsection{Comparison with Classical Sampling}\label{sec:comparison}

On a classical computer, the Klein sampler can be used to sample for large values of  $\sigma$ (which depends on the qualitiy of the basis) in polynomial time. For smaller values of $ \sigma$, the Markov Chain Monte Carlo (MCMC) algorithm~\cite{MTMK} allows to sample at the expense of a larger time complexity. The cost of approximately sampling from $D_{\L(B),\sigma}$ using this algorithm is
\begin{equation*}
    \poly(m) \times \nu\times \frac{1}{\Delta}, \quad \Delta = \frac{\rho_{\sigma}(\L(B))}{\prod_{i=1}^m \rho_{\sigma/||\widetilde{\vec b}_i||}(\Z)} \enspace,
\end{equation*}
where $2^{-\nu}$ indicates the distance to the discrete Gaussian we target. The $\widetilde{\vec b}_i$ are again the columns of the Gram-Schmidt transform of $B$. 

With~\autoref{theorem: performance of quantum gaussian sampling algorithm}, we proved that we can sample quantumly from a distribution that can be made arbitrarily close to $D_{\L(B),\sigma}$, in time:
\begin{gather*}
    \sqrt{\frac{ C^m\prod_{i=1}^{m} \rho_{\sigma_{i}}(\iintv{-2^M,2^M}) }{ \rho_{\sigma}(\Omega) }} \times  \bigOt{ m M (\nu + mM + m^2 r)^2 },\\
    \text{where }C :=\frac{1}{1-2^M\sqrt{2\pi e}\times \exp{-\pi 2^{2M}}} \enspace.
\end{gather*}
Suppose that $M = km$ where $k > 1$ is a constant. We can notice that, as $m$ grows, $C^m$ is very close to 1, and therefore has a negligible influence on the complexity:
\begin{align*}
C^m =  \left( 1-2^M\sqrt{2\pi e}\times \exp{-\pi 2^{2M}} \right)^{-m} \simeq 1 + m 2^M\sqrt{2\pi e}\times \exp{-\pi 2^{2M}} = 1 + o(1) \enspace.
\end{align*}



Moreover, since $\Omega\supset B\prod_{i=1}^m \iintv{-2^{M-1},2^{M-1}}$, if $M=km$ with a constant $k>1$, we can reuse~\autoref{lemma: inequality on gaussian masses on finite sets} and~\autoref{prop: bound on some dtv and gaussian mass} to show that:
\begin{equation*}
    w = \frac{\prod_{i=1}^{m} \rho_{\sigma_{i}}(\iintv{-2^M,2^M}) }{ \rho_{\sigma}(\Omega) }\simeq \frac{\prod_{i=1}^{m} \rho_{\sigma_{i}}(\Z) }{ \rho_{\sigma}(\L(B)) } = \frac{1}{\Delta} \enspace,
\end{equation*}
up to a polynomial factor. Therefore, if we sample from $D_{\L(B),\sigma}$ using~\autoref{theorem: performance of quantum gaussian sampling algorithm}, we roughly gain a quadratic speedup with respect to the MCMC algorithm, since the complexity was $1 / \Delta$ (up to polynomial factors in $\nu, m$ and $r = \max \mathsf{size}(R_{ij}, c_i', \sigma)$) and becomes $1 / \sqrt{\Delta}$.


\section{Application to Dual Attacks on LWE}\label{sec: Improving the quantum dual attack}

In this section, we apply our quantum discrete Gaussian sampling algorithm to dual attacks on the LWE problem. Recall from~\autoref{sec:LWE-SIS} that the search-LWE problem is to find $\vec s \in \Z_q^n$, given a pair $(A,\vec b)$, where $A$ is a matrix in $\Z_q^{m\times n}$, and $\vec b := A\vec s + \vec e$, with $\vec e$ a small error vector in $\Z_q^m$.

The general idea behind dual attacks goes as follows. One splits the matrix $A$ column-wise, and writes $A = [\Aguess{}||\Adual{}]$. One also writes $\vec s = [\sguess{}|\sdual{}]$ accordingly, so that we have $\vec b =  \Aguess{} \sguess{} + \Adual{} \sdual{} + \vec e \enspace.$
Then, one samples short vectors according to a discrete Gaussian distribution, in the dual of the lattice generated by $\Adual{}$. Finally, one guesses $\sguess{}$. Whether the guess is correct is checked with the help of the short vectors sampled just before.


For the first step, the typical approach is to use sieving algorithms in order to obtain short vectors in the desired lattice~\cite{DBLP:conf/eurocrypt/Albrecht17,DBLP:conf/indocrypt/EspitauJK20}. For the second step, the most efficient approach is to use a Fast Fourier Transform (FFT) to speed up the test of candidates for $\sguess{}$~\cite{DBLP:conf/asiacrypt/GuoJ21,MATZOV,DBLP:journals/iacr/CarrierST22,DBLP:conf/crypto/DucasP23,DBLP:conf/crypto/CarrierMST25}. These two strategies are combined to obtained the current most efficient versions of the attack. While these are classical algorithms, quantum speedups are known for sieving~\cite{DBLP:conf/pqcrypto/LaarhovenMP13,laarhoven2015,DBLP:conf/asiacrypt/ChaillouxL21}, which is an important building blocks of all attacks on LWE, and quantum dual attacks were studied in~\cite{DBLP:journals/iacr/AlbrechtS22}.

However, Ducas and Pulles~\cite{DBLP:conf/crypto/DucasP23} pointed out a problem in several dual attacks using sieving algorithms along with the FFT. One of the ground statistical assumptions supporting the correctness of these dual attacks, is false in a lot of cases. This especially affects~\cite{DBLP:conf/asiacrypt/GuoJ21,MATZOV}, and prompted further scrutiny into the correctness of dual attacks. As a result, Pouly and Shen~\cite{DBLP:conf/eurocrypt/PoulyS24} presented a new statistical analysis to justify that a revisited dual attack works -- which they call a ``modern dual attack'', that allows for more rigorous correctness analysis. They also gave a quantum variant of this modern dual attack.

In this section, we first describe the ``modern dual attack'', both classical and quantum, of Pouly and Shen~\cite{DBLP:conf/eurocrypt/PoulyS24}. We then propose two variants of this quantum algorithm. The first variant naively plugs the quantum Gaussian sampling algorithm of~\autoref{sec: quantum ways of samplig from pi} into the first step (sampling short vectors). This improves the time complexity of the entire attack, but similarly to~\cite{DBLP:conf/eurocrypt/PoulyS24}, requires qRAM in the second step. The second variant combines the quantum Gaussian sampling algorithm with the quantum algorithm for mean estimation of~\autoref{corollary:montanaro bis}, directly into the second step. This algorithm improves over the classical complexity of~\cite{DBLP:conf/eurocrypt/PoulyS24} without using qRAM.


\subsection{The Modern Dual Attack on LWE}

Here we present the idea of the ``modern dual attack'' of Pouly and Shen~\cite{DBLP:conf/eurocrypt/PoulyS24}, first the classical algorithm (\autoref{alg:classical modern dual attack}), then its quantum version. 

\subsubsection{The Classical Version and its Complexity.}

Assume that $(A,\vec b)$ is sampled from $LWE(m,\vec s, \chi_e)$. The secret $\vec s$ is in $\Z_q^n$. Let us write $n = \nguess{} +\ndual{}$. We split $\vec s$ into two parts $\sguess{} \in \Z_q^{\nguess{}}$ and $\sdual{} \in \Z_q^{\ndual{}}$, and we also split $A := [\Aguess{} \| \Adual{} ]$ accordingly. Now, we have $\vec b = \Aguess{} \sguess{} + \Adual{} \sdual{} + \vec e \enspace.$
We guess the coefficients of $\sguess{}$. Let us write this guess $\osguess{}$ to avoid confusion. The intuition that underlies the attack is that the good $\sguess{}$ minimizes $\mathrm{dist}(\vec{b}-\Aguess{} \osguess{},\L_q(\Adual{}))$. The goal is therefore to estimate this distance for all $\osguess{}$ with sufficient precision to spot the best candidate.
Fortunately, the periodic Gaussian function $f_{\L_q(\Adual{}),1/\sigma}(\vec y)$, defined in~\autoref{eq:periodic Gaussian function}, can serve as a measure of the distance of a point $\vec y$ from the lattice $\L_q(\Adual{})$: closer points will have a larger value. More precisely, we have

\begin{restatable}{lemma}{lemdistance}[{\cite[Lemma 8]{DBLP:conf/eurocrypt/PoulyS24}}]\label{lemma:distance from a lattice estimation}
    Consider $\L\subset \R^m$ a lattice and $\sigma >0$. Then, for any $\vec x \in \R^m$
    \begin{itemize}
        \item $f_{\L,1/\sigma}(\vec x)\geqslant \rho_{1/\sigma}(\mathrm{dist}(\vec x,\L))$,
        \item if $\mathrm{dist}(\vec x,\L)\geqslant \tau := \frac{1}{\sigma}\sqrt{m/2\pi}$, then $f_{\L,1/\sigma}(\vec x)\leqslant \rho_{1/\sigma}(\mathrm{dist}(\vec x,\L)-\tau)$.
    \end{itemize}
\end{restatable}

which can be loosely interpreted as:
\begin{equation*}
    f_{\L_q(\Adual{}),1/\sigma}(\vec y) \simeq \rho_{1/\sigma}(\mathrm{dist}(\vec y, \L_q(\Adual{}))) \enspace.
\end{equation*}
To approximate this function, we build a set $W = \{\vec w_1,\dotsc, \vec w_N\}$ of vectors sampled from $D_{\L_q^\perp(\Adual{}),q\sigma}$.
Then, we compute
\begin{equation*}
    g_W(\vec y) := \frac{1}{N}\sum_{j=1}^N \cos(2\pi\langle \vec y,\vec w_j\rangle/q), \ \vec y := \vec b-\Aguess{} \osguess{} \enspace.
\end{equation*}
The function $g_W$ is close to the function $f_{\L_q(\Adual{}),1/\sigma}$ (\autoref{theorem: Poitwise approximation lemma}). This gives the following algorithm (note that a Discrete Fourier Transform is used to accelerate the computation of scores $S$).


\begin{algorithm}[tb]
\caption{Classical modern dual attack~\cite[Algorithm 2]{DBLP:conf/eurocrypt/PoulyS24}}\label{alg:classical modern dual attack}
\begin{algorithmic}[1]
\Statex \textbf{Input:} Integers $m$, $n = \nguess{}+\ndual{}$, $q$, $N\in \mathbb N^*$. A LWE sample $(A,\vec b)$, a list $W = \{\vec w_1,\dotsc, \vec w_N\}$ of vectors sampled from $D_{\L_q^\perp(\Adual{}),q\sigma}$
\smallskip
\Statex \textbf{Output:} A guess of $\nguess{}$ coordinates of $\vec s$, or $\bot$
\medskip
\State $\sguess{}\leftarrow \perp$
\State  $S_{max} \leftarrow 0$

\For{$\osguess{}\in \Z_q^{\nguess{}}$}
     \State Compute each $y_1,\dotsc,y_N$, where $y_j := \vec w_j^T(b-\Aguess{} \osguess{})$
      \State $S\leftarrow \sum_{j=1}^N \cos(2\pi y_j/q)$
       \If{$S\geqslant S_{max}$}
           \State $S_{max}\leftarrow S$
            \State $\sguess{} \leftarrow \osguess{}$\
        \EndIf
\EndFor
\State \textbf{Return} $\sguess{}$
\end{algorithmic}
\end{algorithm}

\begin{theorem}[{\cite[Theorem 6]{DBLP:conf/eurocrypt/PoulyS24}}]
    Let $A\in \Z_q^{m\times n}$, $\vec e\in \Z^m$, $\vec s\in \Z_q^n$, $\sigma,\delta>0$, and $N\in \mathbb N$. Let $\tau := \frac{1}{\sigma}\sqrt{m/2\pi}$. Assume that $m\geqslant n$, $A$ has full rank, $\lambda_1(\L_q(A))\geqslant \tau + \|e\|$, and 
    \begin{equation*}
        \rho_{1/\sigma}(\vec e)-\rho_{1/\sigma}(\lambda_1(\L_q(A))- \|e \|-\tau)>2\delta.
    \end{equation*}
    Let $\vec b = A\vec s + \vec e \mod q$. Let $W = \{\vec w_1,\dotsc,\vec w_N\}$ be vectors sampled from $D_{\L_q^\perp(\Adual{}),q\sigma}$. Then,~\autoref{alg:classical modern dual attack} with inputs $(m,n=\nguess{}+\ndual{},q,N,(A,\vec b),W)$ runs in time 
    \begin{equation*}
        \poly(m,n)\times (N+q^{\nguess{}})
    \end{equation*}
    and returns $\sguess{}$ with probability at least $1-q^m\times 2^{-\Omega(N\delta^2)}$ over the choice of $W$.
\end{theorem}

\subsubsection{The Quantum Version and its Complexity.}

The algorithm presented above is an exhaustive search over the set of possible $\osguess{}$, where we look for $\osguess{}$ maximizing the score $S$. The quantum version uses two tools. The first is an approximation of the score $S$ (\autoref{th:algo to compute gW quantumly} below).

\begin{theorem}[{\cite[Theorem 5]{DBLP:journals/iacr/AlbrechtS22},\cite[Theorem 9]{DBLP:conf/eurocrypt/PoulyS24}}]\label{th:algo to compute gW quantumly}
    Let $N\in \mathbb N^*$, and let $W = \{\vec w_1,\dotsc,\vec w_N\}\subset \Z^m$. Let $g_W(\vec y) := \frac{1}{N}\sum_{i = 1}^N \cos(2\pi\langle \vec w_i,\vec y\rangle/q )$. Define $O_W$ to be an oracle such that
    \begin{equation}\label{eq: OW oracle definition}
        O_W : \ket j \ket 0 \to \ket j \ket{\vec w_j}.
    \end{equation}
    For any $\epsilon,\delta >0$, there exists a quantum algorithm $\mathcal A$ that, given $\vec y$ and oracle access to $O_W$, outputs $\mathcal A^{O_W}(\vec y)$ such that $|\mathcal A^{O_W}(\vec y)-g_W(\vec y)|\leqslant \epsilon$, with probability $1-\delta$. The algorithm makes $O(\epsilon^{-1}\log(1/\delta))$ queries to $O_W$, and requires $O(\log(1/\epsilon)+\poly(\log(m)))$ qubits.
\end{theorem}

The second tool is a quantum maximum-finding algorithm that is robust to errors (since the computation of $S$ makes such errors). For this purpose, we can directly use~\autoref{theorem: cost of finding the maximum with faulty function} (generalization of~\cite[Theorem 8]{DBLP:conf/eurocrypt/PoulyS24}) with a constant probability of error. The combination of these two tools gives~\autoref{alg:quantum modern dual attack}.

\begin{algorithm}[tb]
    \caption{Quantum modern dual attack}\label{alg:quantum modern dual attack}
\begin{algorithmic}[1]
\Statex \textbf{Input:} Integers $m$, $n = \nguess{}+\ndual{}$, $q$, $N\in \mathbb N^*$. A real number $\eta>0$. A LWE sample $(A,\vec b)$, an oracle $O_W$ for a set $W = \{\vec w_1,\dotsc, \vec w_N\}$ of vectors sampled from $D_{\L_q^\perp(\Adual{}),q\sigma}$
\smallskip
\Statex \textbf{Output:} A guess of $\nguess{}$ coordinates of $\vec s$, of nothing ($\perp$).
\medskip
\State Set $\A$ as in planned by~\autoref{th:algo to compute gW quantumly}, with $\delta = 1/10$, $\epsilon = \eta$, and $q$
\State $\ket s \leftarrow \sum_{\osguess{}\in \Z_q^{\nguess{}}}  \ket{\osguess{}} \ket{\mathcal A^{O_W}(\vec b - \Aguess{}\osguess{})}$
\State Run the algorithm of~\autoref{theorem: cost of finding the maximum with faulty function} to find $\osguess{}$ such that $\mathcal A^{O_W}(\vec b - \Aguess{}\osguess{})$ is maximal
\State \textbf{Return} $\sguess{} = \osguess{}$
\end{algorithmic}
\end{algorithm}

\begin{theorem}[{\cite[Theorem 9]{DBLP:conf/eurocrypt/PoulyS24}}]\label{theorem: performances of the previous modern dual attack}
    Let $A\in \Z_q^{m\times n}$, $\vec e\in \Z^m$, $\vec s\in \Z_q^n$, $\sigma,\delta,\eta>0$, and $N\in \mathbb N$. Let $\tau := \frac{1}{\sigma}\sqrt{m/2\pi}$. Assume that $m\geqslant n$, $A$ has full rank, $\lambda_1(\L_q(A))\geqslant \tau + ||e||$, and 
    \begin{equation*}
        \rho_{1/\sigma}(\vec e)-\rho_{1/\sigma}(\lambda_1(\L_q(A))-||e||-\tau)>2\delta+\eta.
    \end{equation*}
    Let $\vec b = A\vec s + \vec e \mod q$. Let $W = \{\vec w_1,\dotsc,\vec w_N\}$ be vectors sampled from $D_{\L_q^\perp(\Adual{}),q\sigma}$, and let $O_W$ be an oracle for $W$ as defined in~\autoref{eq: OW oracle definition}. Then,~\autoref{alg:quantum modern dual attack} with inputs $(m,n=\nguess{}+\ndual{},q,N,\eta/2,(A,\vec b),W,O_W)$ makes $O(\eta^{-1}\times q^{\nguess{}/2})$ calls to $O_W$.
    It has a qubit cost in $O(\log(\eta^{-1})+\poly(\log(N)))$ and returns $\sguess{}$ with probability at least $1-q^m\times 2^{-\Omega(N\delta^2)}$.
\end{theorem}

The cost of sampling the necessary $\vec w_j$ should also be included in the cost of the attack. Pouly and Shen~\cite{DBLP:conf/eurocrypt/PoulyS24} proposed to do so by computing a BKZ-reduced basis $B$ of $\L_q^\perp(\Adual{})$, and using the Markov Chain Monte Carlo (MCMC) algorithm $N$ times to sample $N$ vectors in $\L_q^\perp(\Adual{})$~\cite{MTMK}. As recalled in~\autoref{sec:comparison}, the cost of the MCMC algorithm is 
\begin{equation*}
N \times \poly(n) \times \nu\times \frac{1}{\Delta}, \quad \Delta = \frac{\rho_{q\sigma}(\L_q^\perp(\Adual{}))}{\prod_{i=1}^m \rho_{q\sigma/||\widetilde{\vec b}_i||}(\Z)} \enspace,
\end{equation*}
where $2^{-\nu}$ is the distance between the distribution $D_{\L_q^\perp(\Adual{}),q\sigma}$, and the one we actually sample from using the MCMC algorithm.

Pouly and Shen show (\cite[Section 4.3]{DBLP:conf/eurocrypt/PoulyS24}) that one can take $\eta = \delta / 100$ and $\eta^{-1} = \poly(\log m)) \sqrt{N}$. They propose two methods for the quantum attack.
\begin{itemize}
\item Without qRAM: preprocessing the basis with BKZ and directly using the sampler in the algorithm $O_W$ (instead of using precomputed data). The preprocessing step admits a quantum speedup\footnote{BKZ relies on solving the shortest vector problem (SVP) in smaller dimensions. This can be done either by sieving or enumeration, both of which admit a quantum speedup \cite{DBLP:conf/asiacrypt/AonoNS18,DBLP:conf/eurocrypt/BonnetainCSS23}. In most crytographic applications, lattice sieving is always faster than enumeration.}. We denotes its time by $T_{\text{Q-BKZ}} $.
The quantum time complexity of the attack becomes:
\begin{equation}\label{eq:ps-no-qram}
T_{\text{Q-BKZ}} + \frac{1}{\Delta} \sqrt{N} q^{\nguess{}/2} \poly(\log m) \enspace.
\end{equation}
They note that there is no quantum speedup on the sampler (whether one uses Klein's algorithm directly or MCMC). 

\item With qRAM: using the same sampler to precompute the $\vec{w}_1, \ldots, \vec{w}_N$, they need to be stored in classical memory with quantum random access (QRACM) in order to implement the $O_W$ algorithm efficiently. The quantum time complexity becomes:
\begin{equation}\label{eq:ps-qram}
T_{\text{Q-BKZ}} + \frac{1}{\Delta} N \poly(\log m) +  \sqrt{N} q^{\nguess{}/2} \poly(\log m) \enspace.
\end{equation}
\end{itemize}


\subsection{Improvement with qRAM}



A first simple way to improve the quantum modern dual attack of~\cite{DBLP:conf/eurocrypt/PoulyS24}, is to sample the $N$ vectors in $\L_q^\perp(\Adual{})$ using the algorithm of~\autoref{theorem: performance of quantum gaussian sampling algorithm}, instead of the MCMC algorithm. As we have shown in~\autoref{sec:gaussian-distrib}, we gain a quadratic speedup on the MCMC sampler (up to polynomial factors). As a consequence, the quantum time complexity becomes:
\begin{equation}
T_{\text{Q-BKZ}} + \frac{1}{\sqrt{\Delta}} N \poly(\log m) +  \sqrt{N} q^{\nguess{}/2} \poly(\log m) \enspace,
\end{equation}
which is always better than~\autoref{eq:ps-qram}. The time complexity of the second step (maximum-finding) remains unchanged, as well as the quantum memory complexity $N$.

We do not provide concrete numbers for the attack on the Kyber cryptosystem because our analysis, which follows that of \cite{DBLP:conf/eurocrypt/PoulyS24} does not include modulus switching which is required to obtain a competitive complexity.

%
%

\subsection{Improvement without qRAM}


Denote by $B$ a basis of the lattice $\L_q^\perp(\Adual{})$, as before, and a finite set $\Omega\subset \L(B)$. Now, instead of classically sampling $N$ short vectors according to a distribution $D_{\L_q^\perp(\Adual{}),q\sigma}$, as in~\cite{DBLP:conf/eurocrypt/PoulyS24}, we use the algorithm of~\autoref{theorem: performance of quantum gaussian sampling algorithm}, along with~\autoref{corollary:montanaro bis}, to estimate the value $f_{\hat\L_q(\Adual{}),1/q\sigma}$. The resulting dual attack on LWE is presented in~\autoref{alg:new quantum modern dual attack}. Recall that our approximate distribution $D_{\Omega, q\sigma}$ replaces the distribution $D_{\L_q^\perp(\Adual{}),q\sigma}$ (\autoref{lemma:distance from a lattice estimation}).

\begin{algorithm}[tb]
    \caption{New quantum modern dual attack}\label{alg:new quantum modern dual attack}
\begin{algorithmic}[1]
\Statex \textbf{Input:} Integers $m$, $M>1$, $n = \nguess{}+n_{dual}$, $q$. Real numbers $\epsilon,\eta_1,\eta_2>0$, $\sigma>0$. A LWE sample $(A,\vec{b})$.
\smallskip
\Statex \textbf{Output: } A guess of $\nguess{}$ coordinates of $\vec s$.
\medskip

\State Set $\phi$ as the quantum algorithm $:\ket{\vec x}\ket{\vec w} \to \ket{\cos(2\pi\langle \vec x,\vec w\rangle/q)}$

\State Set $\mathcal{GS}$ as the sampling algorithm resulting from~\autoref{theorem: performance of quantum gaussian sampling algorithm}, which performs: $\ket{0} \mapsto \ket{D}$ where $D$ is a distribution close to $D_{\Omega,q\sigma}$



\State Set $\A^{\mathcal{GS},\phi}$ the mean estimation algorithm of~\autoref{corollary:montanaro bis} with parameters $\epsilon$ and $\delta=1/(20\times q^m)$ which performs: $\ket{\vec{x}} \ket{0} \mapsto \ket{\vec{x}} \ket{\A^{\mathcal{GS},\phi}(\vec{x})}$ with $\phi$ the function $\vec{w} \to \cos(2\pi\langle \vec x,\vec w\rangle)$ and $\mathcal{C}$ the circuit $\mathcal{GS}$



\State Run the algorithm of~\autoref{theorem: cost of finding the maximum with faulty function} to find $\overline{s}_{guess}$ such that $\mathcal \A^{\mathcal{GS},\phi}(\vec b - \Aguess{}\overline{s}_{guess})$ is maximal

\State \textbf{Return} $\sguess{} = \overline{s}_{guess}$
\end{algorithmic}
\end{algorithm}

\begin{restatable}{lemma}{lembound}\label{lemma: third distance bound between gw and f L,sigma}
    Let $B\in \Z_q^{m\times n}$, $\L\subset \R^m$ a lattice, $\vec e\in \Z^m$, $\epsilon,\sigma,\delta,\eta_1>0$, and $N\in \N$. Let $\tau := \frac{1}{\sigma}\sqrt{m/2\pi}$, $\eta_2\geqslant 0$, and assume that $\lambda_1(\L+\L_q(B))\geqslant \tau + ||\vec e||$, and 
    \begin{equation*}
        \rho_{1/\sigma}(\vec e) - \rho_{1/\sigma}(\lambda_1(\L+\L_q(B))-||\vec e||-\tau)>2\epsilon+2\eta_1 +\eta_2.
    \end{equation*}
    Let $D$ be a probability distribution over a finite subset $\Omega$ of $\L_q^\perp(B)$, such that $d_{TV}(D,D_{\L_q^\perp(B),q\sigma})\leqslant \eta_1/2$. Let $\mathcal{GS}$ be a quantum circuit such that $\mathcal{GS} : \ket 0 \to \ket{D}$, and let $\phi:\ket{\vec x}\ket{\vec w} \to \ket{\cos(2\pi\langle \vec x,\vec w\rangle/q)}$.  Denote by $\widetilde f(\vec x)$ the output of $\A^{\mathcal{GS},\phi}$ on input $\vec x$ and with parameters $\epsilon$ and $\delta$.

    Then, with probability at least $1-q^m\delta$, for all $\vec x\in \L\backslash \L_q(B)$, we have 
    \begin{gather*}
        |\widetilde f(\vec x) - f_{\L_q(B),1/\sigma}(\vec x)|\leqslant \epsilon + \eta_1,\\
        \widetilde f(\vec e)\geqslant \rho_{1/\sigma}(\vec e)-\epsilon-\eta_1 >\rho_{1/\sigma}(\lambda_1(\L+\L_q(B))-||\vec e||-\tau)+\epsilon+\eta_1 +\eta_2\geqslant \widetilde f(\vec e+\vec x)+\eta_2\enspace .
    \end{gather*}
\end{restatable}

The proof is given in~\autoref{sec: Proof of some results}. 

We now have the necessary tools to prove the following theorem.

\begin{theorem}\label{theorem: performance of main algorithm}
    Consider a matrix $A\in \Z_q^{m\times n}$, vectors $\vec e\in \Z^m$, $\vec s\in \Z_q^n$, real numbers $\epsilon,\sigma,\eta_1,\eta_2>0$, and $\delta=1/(20\times q^m)$. Let $\tau := \frac{1}{\sigma}\sqrt{m/2\pi}$ and denote by $\L$ the lattice $\L_q(A)$. Assume that $m\geqslant n$, $A$ has full rank, $\lambda_1(\L_q(A))\geqslant \tau + ||e||$, and 
    \begin{equation*}
        \rho_{1/\sigma}(\vec e)-\rho_{1/\sigma}(\lambda_1(\L_q(A))-||e||-\tau)>2\epsilon+2\eta_1 +\eta_2.
    \end{equation*}

    Let $\vec b = A\vec s + \vec e \mod q$. Let $D$ be a probability distribution over a finite subset $\Omega$ of $\L_q^\perp(\Adual{})$, such that $d_{TV}(D,D_{\L_q^\perp(\Adual{}),q\sigma})\leqslant \eta_1/2$. 
    Let $\mathcal{GS}$ and $\phi$ be the quantum circuits as defined in~\autoref{alg:new quantum modern dual attack}. 
    
    Then,~\autoref{alg:new quantum modern dual attack} with inputs $(m,M,n=\nguess{}+n_{dual},q,\epsilon,\eta_1,\eta_2,\sigma,(A,\vec b))$ makes
    \begin{equation*}
        O\left( \poly(\log q,m, \log (1/\epsilon)) \sqrt{q^{\nguess}}\right),
    \end{equation*}
    calls to $\phi$, and 
    \begin{equation*}
        O\left( \poly(\log q,m, \log (1/\epsilon)) \sqrt{q^{\nguess}}\times \frac{1}{\epsilon} \right)
    \end{equation*}
    calls to $\mathcal{GS}$. It returns $\sguess{}$ with probability at least $15/16$. It has an additional qubit count in $O(\log(1/\epsilon))+\poly(\log q,m)$.
\end{theorem}


\begin{proof}
    \textbf{Correctness.} Recall that $\widetilde f(\vec x)$ is the output of the probabilistic algorithm $\mathcal A^{\mathcal{GS},\phi}$ on input $\vec x$ and with parameters $\epsilon$ and $\delta$. This output is an estimation of $f_{\L_q(\Adual{}),1/q\sigma}(\vec x)$. Assume that $\delta = 1/(20\times q^m)$. In this case, by~\autoref{lemma: third distance bound between gw and f L,sigma}, with probability at least $1-1/20$, we have for all $\vec x\in \L\backslash \L_q(\Adual{})$,
    \begin{equation}\label{eq: basic gW inequality}
        \widetilde f(\vec e) > \widetilde f(\vec e+\vec x)+\eta_2.
    \end{equation} 
    We say that in this case, $\mathcal A^{\mathcal{GS},\phi}$ returns a correct value, and otherwise it returns a faulty value.
We can then apply~\autoref{theorem: cost of finding the maximum with faulty function} to say that, with probability at least $15/16$, the output $\overline{s}'_{guess}$ of~\autoref{alg:new quantum modern dual attack} maximizes $\widetilde f(\vec b-\Aguess{}\overline{s}_{guess})$.

    Let us now show that, in this case, $\overline{s}'_{guess} = \sguess{}$. Assume that this is not true, for the sake of contradiction. 

    Note that $A$ has full rank, and $m\geqslant n$, so the columns of $A$ are linearly independent over $\Z_q$ and 
    \begin{equation*}
        \L\backslash\L_q(\Adual{}) = \L_q(\Aguess{})\backslash \L_q(\Adual{}) = \L_q(\Aguess{})\backslash q\Z^m.
    \end{equation*}
    Moreover, since $\vec b = A\vec s+\vec e$, we have
    \begin{equation*}
        \vec b-\Aguess{}\overline{s}'_{guess} = \Aguess{}(\sguess{}-\overline{s}'_{guess})+\vec e \enspace.
    \end{equation*}
    So, if $\overline{s}'_{guess}\neq \sguess{}$, there is some $\vec x\in \L_q(\Aguess{})$ such that $\vec b-\Aguess{}\overline{s}'_{guess} = \vec e+\vec x$. Here,~\autoref{lemma: third distance bound between gw and f L,sigma} applies. Namely, with probability $1-q^m\delta$, we have
    \begin{equation}
        \widetilde f(\vec e) > \widetilde f(\vec e+\vec x)+\eta_2.
    \end{equation} 
    This is in contradiction with the fact that $\overline{s}'_{guess}$ maximizes $\widetilde f(\vec b-\Aguess{}\overline{s}_{guess})$. Therefore, $\overline{s}'_{guess} = \sguess{}$.

    \textbf{Performances.} Computing $\vec b-\Aguess{}\overline{s}_{guess}$ on a quantum computer involves products of $m\times \nguess{}$ dimensional matrices with integers modulo $q$, so this can be done with a gate and qubit count in $\poly(\log q,m)$.

    
    By~\autoref{corollary:montanaro bis}, and since $\delta = 1/(20\times q^m)$, the operator $\mathcal A^{\mathcal{GS},\phi}$ calls $\phi$ a number 
     $\bigOt{m\log(q)\log\left(\frac{1}{\epsilon}\right)}$ 
    of times. Likewise, it calls $\mathcal{GS}$ a number of times in
         $\bigOt{m\log(q)\frac{1}{\epsilon}}$.
    The qubit count is in $O(\log(1/\epsilon)+m\log(q))$.

    By~\autoref{theorem: cost of finding the maximum with faulty function}, to find the $\overline{s}_{guess}$ that maximises $\widetilde f(\vec b-\Aguess{}\overline{s}_{guess})$, we need to call $\mathcal A^{\mathcal{GS},\phi}$ a number of times in $O(\nguess{}\log(q)\sqrt{q^{\nguess{}}})$, and to use $O(\nguess{}\log(q)\sqrt{q^{\nguess{}}})$ additional gates, and $\poly(\nguess{}\log(q))$ qubits. By multiplying this (and using $\nguess{}<m$) we obtain the total number of calls to $\mathcal{GS}$ and $\phi$,
with an additional gate count accounting for a polynomial factor.
\qed
\end{proof}

\paragraph{Concrete Complexity and Comparison.}
Similarly to~\cite[Section 4.3]{DBLP:conf/eurocrypt/PoulyS24}, we can take a concrete value for $\epsilon$ and express the complexity in terms of the parameter $N$ of the previous attack. Indeed, one should choose $\epsilon$ to be smaller than the smallest possible value of $\rho_{1/\sigma}(\vec{e})$, and this corresponds (up to a polynomial factor) to $\epsilon = \frac{1}{\sqrt{N}}$ where $N$ is the amount of samples required in the first variant of the attack.

Therefore, by combining our new sampling with~\autoref{theorem: performance of main algorithm} we obtain a complexity:
\begin{equation}
T_{\text{Q-BKZ}} + \frac{1}{\sqrt{\Delta}} \sqrt{N} q^{\nguess{}/2} \poly(m, \log q) \enspace,
\end{equation}
which improves by a factor $\sqrt{\Delta}$ over~\autoref{eq:ps-no-qram}.

\paragraph{Comparison with sieving-based dual attacks.}
The main advantage of this improvement is that it does not require a qRAM
and only a polynomial number of qubits.
This should be compared to dual attacks based on sieving which,
although more efficient in the classical setting\footnote{Sieving based dual attacks are more efficient but require heuristic assumptions \cite{DBLP:conf/crypto/CarrierMST25}, whereas the "modern" dual attack in \cite{DBLP:conf/eurocrypt/PoulyS24} using a discrete gaussian sampler is provably correct.}, would require very large qRAMs in
the quantum setting.

\section{Quantum Algorithm for solving SIS in any norm}\label{sec: Improving SIS}
In \cite{PS26}, the authors gave a simple algorithm that solves the $\SIS$ problem for any norm. The generic framework of the algorithm is the following: given access to a discrete Gaussian sampler, they estimate the probability $p$ for the sampler to produce a short enough vector that corresponds to the solution of the $\SIS$ problem, using averaging results on $q$-ary lattices \cite{DBLP:conf/eurocrypt/PoulyS24}. They then sample enough vectors so that the probability that one of them is the solution of the $\SIS$ problem is constant. The complexity of the algorithm is then $O(1/p\cdot \text{T}_{\text{sampler}})$ where $\text{T}_{\text{sampler}}$ is the complexity
of producing one sample. Note that this ignores any potential preprocessing required for the sampler, such as reducing the basis.

The discrete Gaussian sampler that they used was precisely the MCMC sampler in \cite{MTMK} for which we have a quantum quadratic speedup (see~\autoref{sec:comparison}).
Note that the overall attack starts by reducing the basis using the BKZ algorithm. This
preprocessing step also admits a quantum speedup, albeit far from a quadratic one when we use lattice sieving for solving the shortest vector problem in smaller dimensions \cite{DBLP:conf/eurocrypt/BonnetainCSS23}.

In our quantum algorithm, we constructed a quantum state $\sum_{\vec{x} \in \Omega} \sqrt{D(\vec{x})} \Ket{\vec{x}}$
where $D$ is (very close to) the discrete Gaussian over the lattice
$L$ of interest, and $\Omega$ is a large domain. From this state and
given $\ell$, the length parameter of the SIS problem, we can easily build
\[
\Ket{\varphi}:=\sum_{\vec{x} \in \Omega} \sqrt{D(\vec{x})} \Ket{\vec{x}}\Ket{f(\vec{x})}.
\]
where $f$ is the indicator function of the $\ell_p$-ball of radius $\ell$
minus the zero vector.
The crux of the analysis in \cite{PS26} lies in the analysis of the probability $p$
that a vector $\vec{x}$ sampled from $D$ falls into this ball. In other words,
\begin{equation}\label{eq:aaprob}
	p=\sum_{\vec{x} \in \Omega} D(\vec{x})f(\vec{x}).
\end{equation}
By doing an amplitude amplification on $\Ket{\varphi}$, we can keep
only\footnote{Technically, we only obtain a state very close to the wanted one
	which means that the final measurement only gives a solution with constant
	probability.
} the
vectors $\vec{x}$ such that $f(\vec{x})=1$. The remaining state will exactly 
contain all vectors in the ball minus the zero vector, so we can measure it
and obtain a solution
to the SIS problem. The complexity of this quantum algorithm is therefore
$O(\text{T}_{\text{Q-sample}}\text{T}_{\text{QAA}})$
where $\text{T}_{\text{Q-sample}}$ is the time complexity of the quantum sampler 
and $\text{T}_{\text{QAA}}$ the time of the quantum amplitude amplification.
We have shown in Section \ref{sec:comparison} that
$\text{T}_{\text{Q-sample}}=\poly(m)\sqrt{\text{T}_{\text{sample}}}$.
Finally, by \eqref{eq:aaprob}, we have $\text{T}_{\text{QAA}}=O(\frac{1}{\sqrt{p}})$.
Hence, the overall complexity of our quantum algorithm is
\[
	T_{\text{attack}}=\text{T}_{\text{Q-BKZ}}+\poly(m)\sqrt{\frac{\text{T}_{\text{sample}}}{p}} \enspace.
\]

In~\autoref{fig:dilithium_classical}, we show the estimated classical time cost of the attack on $\SIS^\infty$ for the Dilithium signature scheme \cite{Ducas2017CRYSTALSDilithiumAS} given in \cite{PS26}.
We modified the estimator of \cite{PS26} to use the quantum complexity
above (neglecting polynomial factors). It uses \cite{DBLP:conf/pqcrypto/PoulyS25} to estimate the complexity of the sampler
for a BKZ reduced basis and optimizes the choice of $\beta$, the BKZ block size, and $s$ to find the optimal set of parameters. We applied this modified estimator 
to the parameters of $\text{SIS}^\infty$ for Dilithium. The results can be found in~\autoref{fig:dilithium}.
We can see that, while the complexity is much better than the classical attack, it remains limited by the BKZ step. It should also be noted that the method for solving $\text{SIS}^\infty$ in \cite{PS26} is provable but still quite far from the heuristic
estimates underlying the security of Dilithium.

\begin{table}[H]
	\centering
	\caption{\label{fig:dilithium_classical}Estimated classical time cost of the attack on $\SIS^\infty$ for Dilithium in \cite{PS26}. The $T_{\text{sample}}$ column corresponds to
		the cost of the MCMC sampler (to produce one sample) and $T_{\text{BKZ}}$ is the cost of
		reducing the basis. $N$ is the number of Gaussian samples used by the attack. Columns $\beta$ and $s$ give the parameters of the attacks (BKZ block size
		and width of the Gaussian sampler).}
	\begin{tabular}{c@{\hspace{4ex}}c@{\hspace{3ex}}c@{\hspace{3ex}}c@{\hspace{3ex}}c@{\hspace{3ex}}c@{\hspace{3ex}}c}
		\hline
		Scheme       &   $\log_2 T_{\text{attack}}$ &   $\beta$ &    $s$ & $\log_2 N$ & $\log_2 T_{\text{sample}}$ & $\log_2 T_{\text{BKZ}}$ \\
		\hline
		NIST Level 2 &    202 &       609 & 346702 &           39 &      159 &   201 \\
		NIST Level 3 &    289 &       919 & 723749 &           55 &      230 &   288 \\
		NIST Level 5 &    400 &      1314 & 769533 &           74 &      323 &   400 \\
		\hline
	\end{tabular}
\end{table}

\begin{table}[htbp]
\centering
\caption{\label{fig:dilithium}Estimated quantum time cost of our attack on $\SIS^\infty$ for Dilithium.
$T_{\text{Q-sample}}$ is the cost of our quantum sampler and $T_{\text{Q-BKZ}}$ is the cost of reducing the basis using quantum lattice sieving algorithm in \cite{DBLP:conf/eurocrypt/BonnetainCSS23}. 
Columns $\beta$ and $s$ give the parameters of the attacks (BKZ block size
and width of the Gaussian sampler).}
\begin{tabular}{c@{\hspace{4ex}}c@{\hspace{3ex}}c@{\hspace{3ex}}c@{\hspace{3ex}}c@{\hspace{3ex}}c@{\hspace{3ex}}c}
	\hline
	Scheme       &   $\log_2 T_{\text{attack}}$ &   $\beta$ &    $s$ & $-\log_2 p$ & $\log_2 T_{\text{Q-sample}}$ & $\log_2 T_{\text{Q-BKZ}}$ \\
	\hline
	NIST Level 2 &    146 &       571 & 358953 &           50 &      119 &   146 \\
	NIST Level 3 &    219 &       854 & 752706 &           72 &      179 &   219 \\
	NIST Level 5 &    312 &      1215 & 804127 &          100 &      259 &   312 \\
	\hline
\end{tabular}
\end{table}

\ifeprint
\subsubsection*{Acknowledgments.}
The authors thank the anonymous reviewers of CRYPTO 2026 for their comments which helped improve the paper. This work has been supported by the French Agence Nationale de la Recherche through the France 2030 program under grant agreement No. ANR-22-PETQ-0008 PQ-TLS.

\fi

\ifeprint
\else
\newpage
\fi

\bibliographystyle{splncs04}
\bibliography{biblio.bib}

\begin{thebibliography}{10}
\providecommand{\url}[1]{\texttt{#1}}
\providecommand{\urlprefix}{URL }
\providecommand{\doi}[1]{https://doi.org/#1}

\bibitem{DBLP:conf/stacs/AggarwalCKS21}
Aggarwal, D., Chen, Y., Kumar, R., Shen, Y.: Improved (provable) algorithms for
  the shortest vector problem via bounded distance decoding. In: {STACS}.
  LIPIcs, vol.~187, pp. 4:1--4:20. Schloss Dagstuhl - Leibniz-Zentrum f{\"{u}}r
  Informatik (2021). \doi{10.4230/LIPICS.STACS.2021.4}

\bibitem{DBLP:journals/siamcomp/AggarwalCKS25}
Aggarwal, D., Chen, Y., Kumar, R., Shen, Y.: Improved classical and quantum
  algorithms for the shortest vector problem via bounded distance decoding.
  {SIAM} J. Comput.  \textbf{54}(2),  233--278 (2025).
  \doi{10.1137/22M1486959}, \url{https://doi.org/10.1137/22m1486959}

\bibitem{DBLP:conf/stoc/AggarwalDRS15}
Aggarwal, D., Dadush, D., Regev, O., Stephens{-}Davidowitz, N.: Solving the
  shortest vector problem in $2^n$ time using discrete gaussian sampling:
  Extended abstract. In: {STOC}. pp. 733--742. {ACM} (2015).
  \doi{10.1145/2746539.2746606}

\bibitem{DBLP:journals/jacm/AharonovR05}
Aharonov, D., Regev, O.: Lattice problems in {NP} cap conp. J. {ACM}
  \textbf{52}(5),  749--765 (2005). \doi{10.1145/1089023.1089025},
  \url{https://doi.org/10.1145/1089023.1089025}

\bibitem{Ajtai96}
Ajtai, M.: Generating hard instances of lattice problems (extended abstract).
  In: {STOC}. pp. 99--108. {ACM} (1996). \doi{10.1145/237814.237838}

\bibitem{DBLP:conf/eurocrypt/Albrecht17}
Albrecht, M.R.: On dual lattice attacks against small-secret {LWE} and
  parameter choices in helib and {SEAL}. In: {EUROCRYPT} {(2)}. Lecture Notes
  in Computer Science, vol. 10211, pp. 103--129 (2017).
  \doi{10.1007/978-3-319-56614-6\_4}

\bibitem{DBLP:conf/asiacrypt/AlbrechtGPS20}
Albrecht, M.R., Gheorghiu, V., Postlethwaite, E.W., Schanck, J.M.: Estimating
  quantum speedups for lattice sieves. In: {ASIACRYPT} {(2)}. pp. 583--613.
  Lecture Notes in Computer Science, Springer (2020).
  \doi{10.1007/978-3-030-64834-3\_20}

\bibitem{DBLP:journals/iacr/AlbrechtS22}
Albrecht, M.R., Shen, Y.: Quantum augmented dual attack. {IACR} Cryptol. ePrint
  Arch. p.~656 (2022), \url{https://eprint.iacr.org/2022/656}

\bibitem{DBLP:conf/asiacrypt/AonoNS18}
Aono, Y., Nguyen, P.Q., Shen, Y.: Quantum lattice enumeration and tweaking
  discrete pruning. In: {ASIACRYPT} {(1)}. Lecture Notes in Computer Science,
  vol. 11272, pp. 405--434. Springer (2018).
  \doi{10.1007/978-3-030-03326-2\_14}

\bibitem{boer2022random}
de~Boer, K.: Random Walks on Arakelov Class Groups. Ph.D. thesis, Leiden
  University (2022), \url{https://hdl.handle.net/1887/3463719}, phD thesis,
  Institutional Repository of Leiden University

\bibitem{DBLP:journals/iacr/BoerF25}
de~Boer, K., Felderhoff, J.: Quantumly computing s-unit groups in quantified
  polynomial time and space. {IACR} Cryptol. ePrint Arch. p.~1825 (2025),
  \url{https://eprint.iacr.org/2025/1825}

\bibitem{PS26}
Bollauf, M.F., Pouly, A., Shen, Y.: Solving {SIS} in any norm via gaussian
  sampling. Cryptology {ePrint} Archive, Paper 2026/225 (2026),
  \url{https://eprint.iacr.org/2026/225}

\bibitem{DBLP:conf/eurocrypt/BonnetainCSS23}
Bonnetain, X., Chailloux, A., Schrottenloher, A., Shen, Y.: Finding many
  collisions via reusable quantum walks - application to lattice sieving. In:
  {EUROCRYPT} {(5)}. Lecture Notes in Computer Science, vol. 14008, pp.
  221--251. Springer (2023). \doi{10.1007/978-3-031-30589-4\_8}

\bibitem{DBLP:conf/eurosp/BosDKLLSSSS18}
Bos, J.W., Ducas, L., Kiltz, E., Lepoint, T., Lyubashevsky, V., Schanck, J.M.,
  Schwabe, P., Seiler, G., Stehl{\'{e}}, D.: {CRYSTALS} - kyber: {A} cca-secure
  module-lattice-based {KEM}. In: EuroS{\&}P. pp. 353--367. {IEEE} (2018).
  \doi{10.1109/EUROSP.2018.00032}

\bibitem{boyer1998tight}
Boyer, M., Brassard, G., H{\o}yer, P., Tapp, A.: Tight bounds on quantum
  searching. Fortschritte der Physik: Progress of Physics  \textbf{46}(4-5),
  493--505 (1998).
  \doi{10.1002/(SICI)1521-3978(199806)46:4/5<493::AID-PROP493>3.0.CO;2-P}

\bibitem{DBLP:conf/stoc/BrakerskiLPRS13}
Brakerski, Z., Langlois, A., Peikert, C., Regev, O., Stehl{\'{e}}, D.:
  Classical hardness of learning with errors. In: {STOC}. pp. 575--584. {ACM}
  (2013). \doi{10.1145/2488608.2488680}

\bibitem{brassard2002quantum}
Brassard, G., Hoyer, P., Mosca, M., Tapp, A.: Quantum amplitude amplification
  and estimation. Contemporary Mathematics  \textbf{305},  53--74 (2002).
  \doi{10.1090/conm/305/05215}

\bibitem{DBLP:conf/crypto/CarrierMST25}
Carrier, K., Meyer{-}Hilfiger, C., Shen, Y., Tillich, J.: Assessing the impact
  of a variant of matzov's dual attack on kyber. In: {CRYPTO} {(1)}. Lecture
  Notes in Computer Science, vol. 16000, pp. 444--476. Springer (2025).
  \doi{10.1007/978-3-032-01855-7\_15}

\bibitem{DBLP:journals/iacr/CarrierST22}
Carrier, K., Shen, Y., Tillich, J.: Faster dual lattice attacks by using coding
  theory. {IACR} Cryptol. ePrint Arch. p.~1750 (2022)

\bibitem{DBLP:conf/asiacrypt/ChaillouxL21}
Chailloux, A., Loyer, J.: Lattice sieving via quantum random walks. In:
  {ASIACRYPT} {(4)}. Lecture Notes in Computer Science, vol. 13093, pp. 63--91.
  Springer (2021). \doi{10.1007/978-3-030-92068-5\_3}

\bibitem{DBLP:journals/iacr/ChoHKLS24}
Cho, B., Hhan, M., Kim, T., Lee, J., Shen, Y.: Does quantum lattice sieving
  require quantum {RAM}? {IACR} Cryptol. ePrint Arch. p.~1700 (2024),
  \url{https://eprint.iacr.org/2024/1700}

\bibitem{Ducas2017CRYSTALSDilithiumAS}
Ducas, L., Kiltz, E., Lepoint, T., Lyubashevsky, V., Schwabe, P., Seiler, G.,
  Stehl{\'e}, D.: Crystals-dilithium algorithm specifications and supporting
  documentation (2017),
  \url{https://api.semanticscholar.org/CorpusID:198994007}

\bibitem{DBLP:journals/tches/DucasKLLSSS18}
Ducas, L., Kiltz, E., Lepoint, T., Lyubashevsky, V., Schwabe, P., Seiler, G.,
  Stehl{\'{e}}, D.: Crystals-dilithium: {A} lattice-based digital signature
  scheme. {IACR} Trans. Cryptogr. Hardw. Embed. Syst.  \textbf{2018}(1),
  238--268 (2018). \doi{10.13154/TCHES.V2018.I1.238-268},
  \url{https://doi.org/10.13154/tches.v2018.i1.238-268}

\bibitem{DBLP:conf/crypto/DucasP23}
Ducas, L., Pulles, L.N.: Does the dual-sieve attack on learning with errors
  even work? In: {CRYPTO} {(3)}. Lecture Notes in Computer Science, vol. 14083,
  pp. 37--69. Springer (2023). \doi{10.1007/978-3-031-38548-3\_2}

\bibitem{DBLP:journals/corr/quant-ph-9607014}
D{\"{u}}rr, C., H{\o}yer, P.: A quantum algorithm for finding the minimum. CoRR
   \textbf{quant-ph/9607014} (1996),
  \url{http://arxiv.org/abs/quant-ph/9607014}

\bibitem{DBLP:conf/indocrypt/EspitauJK20}
Espitau, T., Joux, A., Kharchenko, N.: On a dual/hybrid approach to small
  secret {LWE} - {A} dual/enumeration technique for learning with errors and
  application to security estimates of {FHE} schemes. In: {INDOCRYPT}. Lecture
  Notes in Computer Science, vol. 12578, pp. 440--462. Springer (2020).
  \doi{10.1007/978-3-030-65277-7\_20}

\bibitem{DBLP:conf/stoc/GentryPV08}
Gentry, C., Peikert, C., Vaikuntanathan, V.: Trapdoors for hard lattices and
  new cryptographic constructions. In: {STOC}. pp. 197--206. {ACM} (2008).
  \doi{10.1145/1374376.1374407}

\bibitem{DBLP:conf/asiacrypt/GuoJ21}
Guo, Q., Johansson, T.: Faster dual lattice attacks for solving {LWE} with
  applications to {CRYSTALS}. In: {ASIACRYPT} {(4)}. Lecture Notes in Computer
  Science, vol. 13093, pp. 33--62. Springer (2021).
  \doi{10.1007/978-3-030-92068-5\_2}

\bibitem{DBLP:conf/icalp/HoyerMW03}
H{\o}yer, P., Mosca, M., de~Wolf, R.: Quantum search on bounded-error inputs.
  In: {ICALP}. Lecture Notes in Computer Science, vol.~2719, pp. 291--299.
  Springer (2003). \doi{10.1007/3-540-45061-0\_25}

\bibitem{10.5555/11534.11537}
Jerrum, M.R., Valiant, L.G., Vazirani, V.V.: Random generation of combinatorial
  structures from a uniform. Theor. Comput. Sci.  \textbf{43}(2–3),
  169–188 (Jul 1986)

\bibitem{kitaev2008wavefunction}
Kitaev, A., Webb, W.A.: Wavefunction preparation and resampling using a quantum
  computer. arXiv preprint arXiv:0801.0342  (2008)

\bibitem{DBLP:conf/soda/Klein00}
Klein, P.N.: Finding the closest lattice vector when it's unusually close. In:
  {SODA}. pp. 937--941. {ACM/SIAM} (2000). \doi{DBLP:conf/soda/Klein00}

\bibitem{knuth2011taocp2}
Knuth, D.E.: The Art of Computer Programming, Volume 2: Seminumerical
  Algorithms. Addison-Wesley, 4 edn. (2011)

\bibitem{laarhoven2015}
Laarhoven, T.: Search problems in cryptography: From fingerprinting to lattice
  sieving. Ph.D. thesis, Eindhoven University of Technology (2015)

\bibitem{DBLP:conf/pqcrypto/LaarhovenMP13}
Laarhoven, T., Mosca, M., van~de Pol, J.: Solving the shortest vector problem
  in lattices faster using quantum search. In: PQCrypto. Lecture Notes in
  Computer Science, vol.~7932, pp. 83--101. Springer (2013).
  \doi{10.1007/978-3-642-38616-9\_6}

\bibitem{MATZOV}
MATZOV: Report on the security of {LWE}: Improved dual lattice attack (Apr
  2022). \doi{10.5281/zenodo.6412487},
  \url{https://doi.org/10.5281/zenodo.6412487}

\bibitem{DBLP:conf/focs/MicciancioR04}
Micciancio, D., Regev, O.: Worst-case to average-case reductions based on
  gaussian measures. vol.~37, pp. 267--302 (2007).
  \doi{10.1137/S0097539705447360},
  \url{https://doi.org/10.1137/S0097539705447360}

\bibitem{DBLP:journals/siamcomp/MicciancioR07}
Micciancio, D., Regev, O.: Worst-case to average-case reductions based on
  gaussian measures. {SIAM} J. Comput.  \textbf{37}(1),  267--302 (2007).
  \doi{10.1137/S0097539705447360},
  \url{https://doi.org/10.1137/S0097539705447360}

\bibitem{Montanaro}
Montanaro, A.: Quantum speedup of monte carlo methods. Proceedings of the Royal
  Society A: Mathematical, Physical and Engineering Sciences
  \textbf{471}(2181),  20150301 (09 2015). \doi{10.1098/rspa.2015.0301},
  \url{https://doi.org/10.1098/rspa.2015.0301}

\bibitem{nielsen2010quantum}
Nielsen, M.A., Chuang, I.L.: Quantum Computation and Quantum Information: 10th
  Anniversary Edition. Cambridge University Press, USA, 10th edn. (2011)

\bibitem{DBLP:journals/toct/OzolsRR13}
Ozols, M., Roetteler, M., Roland, J.: Quantum rejection sampling. {ACM} Trans.
  Comput. Theory  \textbf{5}(3),  11:1--11:33 (2013).
  \doi{10.1145/2493252.2493256}, \url{https://doi.org/10.1145/2493252.2493256}

\bibitem{DBLP:conf/eurocrypt/PoulyS24}
Pouly, A., Shen, Y.: Provable dual attacks on learning with errors. In:
  {EUROCRYPT} {(6)}. Lecture Notes in Computer Science, vol. 14656, pp.
  256--285. Springer (2024). \doi{10.1007/978-3-031-58754-2\_10}

\bibitem{DBLP:conf/pqcrypto/PoulyS25}
Pouly, A., Shen, Y.: Discrete gaussian sampling for {BKZ}-reduced basis. In:
  PQCrypto {(2)}. Lecture Notes in Computer Science, vol. 15578, pp. 63--88.
  Springer (2025). \doi{10.1007/978-3-031-86602-9\_3}

\bibitem{Falcon}
Prest, T., Fouque, P.A., Hoffstein, J., Kirchner, P., Lyubashevsky, V., Pornin,
  T., Ricosset, T., Seiler, G., Whyte, W., Zhang, Z.: Falcon. Post-Quantum
  Cryptography Project of NIST  (2020)

\bibitem{DBLP:conf/asiacrypt/QuX25}
Qu, H., Xu, G.: On the provable dual attack for {LWE} by modulus switching. In:
  {ASIACRYPT} {(3)}. Lecture Notes in Computer Science, vol. 16247, pp. 34--64.
  Springer (2025). \doi{10.1007/978-981-95-5099-9\_2}

\bibitem{DBLP:journals/jacm/Regev09}
Regev, O.: On lattices, learning with errors, random linear codes, and
  cryptography. vol.~56, pp. 34:1--34:40 (2009). \doi{10.1145/1568318.1568324},
  \url{https://doi.org/10.1145/1568318.1568324}

\bibitem{DBLP:conf/stoc/0001S17}
Regev, O., Stephens{-}Davidowitz, N.: A reverse minkowski theorem. In: {STOC}.
  pp. 941--953. {ACM} (2017). \doi{DBLP:conf/stoc/0001S17}

\bibitem{sanders2019black}
Sanders, Y.R., Low, G.H., Scherer, A., Berry, D.W.: Black-box quantum state
  preparation without arithmetic. Phys. Rev. Lett.  \textbf{122},  020502 (Jan
  2019). \doi{10.1103/sanders2019black},
  \url{https://link.aps.org/doi/10.1103/sanders2019black}

\bibitem{MTMK}
Wang, Z., Ling, C.: Lattice gaussian sampling by markov chain monte carlo:
  Bounded distance decoding and trapdoor sampling. {IEEE} Trans. Inf. Theory
  \textbf{65}(6),  3630--3645 (2019). \doi{10.1109/TIT.2019.2901497},
  \url{https://doi.org/10.1109/TIT.2019.2901497}

\bibitem{yoder2014fixed}
Yoder, T.J., Low, G.H., Chuang, I.L.: Fixed-point quantum search with an
  optimal number of queries. Physical review letters  \textbf{113}(21),  210501
  (2014)

\end{thebibliography}

\appendix


\section{Finding the maximum in a register}\label{appendix:Finding the maximum in a register}

In this section, we prove~\autoref{theorem: cost of finding the maximum with faulty function} from~\autoref{sec:Finding the maximum in a register}.

\thmmaxfinding*


This is essentially a ``bounded-error'' version of quantum maximum-finding, which combines the results of~\cite{DBLP:journals/corr/quant-ph-9607014} (for maximum-finding) and~\cite[Section 3]{DBLP:conf/icalp/HoyerMW03} (for bounded-error quantum search).

We start from the quantum minimum-finding from~\cite{DBLP:journals/corr/quant-ph-9607014}, adapted for maximum-finding, assuming for now that the function $s$ is exact. The algorithm runs as follows:
\begin{enumerate}
    \item Choose $j \in \{0,1\}^\ell$ uniformly at random.
	\item Apply quantum exponential search~\cite{boyer1998tight}  over the space $\{0,1\}^\ell$ to find $i \in \{0,1\}^\ell$ such that $s(i) > s(j)$. Let $j \leftarrow i$.
    \item Go back to Step $2.$
\end{enumerate}
One should repeat this loop until ``the total running time is less than $22.5\sqrt{2^\ell}+1.4\log_2(2^\ell)$'', as per~\cite{DBLP:journals/corr/quant-ph-9607014}. Indeed, after some time, there will not be any new $i$ such that $s(i) > s(j)$, and the exponential search step will just run forever (before we stop it).

\begin{theorem}[{\cite[Theorem 1, Lemma 2]{DBLP:journals/corr/quant-ph-9607014}}]\label{theorem:algorithm spotting the max in a register, with measures}
Consider a function $s$ with quantum oracle access.
The algorithm presented in~\cite{DBLP:journals/corr/quant-ph-9607014} finds $i\in \{0,1\}^\ell$ such that $s(i)$ is maximal with probability at least $1/2$, in time (and number of calls to $s$) $O(\sqrt{2^\ell})$.
\end{theorem}

In Step 3, \emph{exponential search} is a sequence of Grover searches with an increasing number of iterations, which allows to produce the solution state in some guaranteed average time even though the probability of success of the search itself is unknown.

Like in~\cite{DBLP:conf/eurocrypt/PoulyS24}, we adapt this algorithm to the case where the function $s$ has a non-zero error rate (and only returns the ``true'' value $s_i$ with probability $\geq 9/10$). The idea is to replace the quantum search steps in the exponential search layers, which are instances of Grover's search, by the bounded-error quantum search presented in~\cite{DBLP:conf/icalp/HoyerMW03}. 


\begin{theorem}[{\cite[Section 3]{DBLP:conf/icalp/HoyerMW03}}]\label{theorem: Grover with errors}
Consider a set of $\{F_i\}_{i\in \{0,1\}^\ell}$ functions, such that each $F_i$ is destined to be evaluated on $i$, and to return a value $f_i$, that equals $0$ or $1$. Consider that each $F_i$ can return a faulty result with probability at most $1/10$, and that there are $t$ solutions. There exists a quantum algorithm that finds a solution in average time $O(\sqrt{2^\ell / t})$.
\end{theorem}


The method proposed in~\cite[Section 3]{DBLP:conf/icalp/HoyerMW03} is a combination of exponential search with their bounded-error search. The bounded-error search consists in nested quantum amplification subroutines, in which one amplifies the elements evaluating to 1, and progressively refines the test by using a majority vote on several instances of $F_i$.

In our case, we can immediately plug in the result of~\autoref{theorem: Grover with errors} into the algorithm above. Indeed, the test $f_i$ is exactly: $s(i) > s(j)$, where $j$ is the current constant. While $s(j)$ is evaluated once before running the search, and so can be evaluated with certainty, the error comes from the evaluation of $s(i)$ in the search. Thus the probability of error for $s(i)$ only needs to be 1/10. This gives the following algorithm.


\begin{algorithm}[H]
\caption{Finding the maximum}\label{alg:finding the maximum in a register}
\begin{algorithmic}[1]
\Statex \textbf{Input: } An efficient randomized function $s$ defined on $\{0,1\}^\ell$, that returns a result $s(i)$ that is supposed to be equal to a value $s_i$, and that is faulty with probability at most $1/10$. An integer $k>0$.
\smallskip
\Statex \textbf{Output: }An index $i\in \{0,1\}^\ell$
\medskip
\State Choose $j \in \{0,1\}^\ell$ uniformly at random
\State $j_{max} \leftarrow j$
\For{$k_1 = 1 \to k$}
\Repeat
\State Apply bounded-error quantum exponential search over the space $\{0,1\}^\ell$ to find $i \in \{0,1\}^\ell$ such that $s(i) > s(j)$.
\Until{the loop has expended $22.5\sqrt{2^\ell}+1.4\log_2(2^\ell)$ time}

     \State If $s(j) >s(j_{max})$, set $j_{max}=j$
\EndFor
\State Return $j_{max}$
\end{algorithmic}
\end{algorithm}

We can now give the proof of~\autoref{theorem: cost of finding the maximum with faulty function}.

\begin{proof}
By~\autoref{theorem: Grover with errors}, the exponential search loop finds a solution in $O(\sqrt{2^\ell / t})$ average time, even under bounded-error inputs. As we said above, by assumption, the probability of error of $s$ is exactly the one required by the theorem.

By~\autoref{theorem:algorithm spotting the max in a register, with measures}, the loop succeeds in finding the maximum in time $\bigO{\sqrt{2^\ell}}$ and with probability of success $1/2$. By repeating the algorithm $k$ times and taking the maximum of the obtained values, we amplify the success probability to $1-1/2^k$. Finally, this algorithm is efficient in space, as we only require $\ell$ qubits to represent $i$, and enough space to compute $s$ in superposition. Since $s$ is assumed to be efficient this gives a $\poly (\ell)$ space. There is an additional count of $O(\ell)$ gates per iterate of Grover's search, which gives a total of $\bigO{k \ell \sqrt{2^\ell}}$ additional gates. \qed
\end{proof}

\section{Proofs of Some Results}\label{sec: Proof of some results}

\subsection{Justification of~\autoref{lemma:rho Z approx}}

\rhoZapprox*

This follows immediately from Lemma 2 in~\cite{DBLP:conf/pqcrypto/PoulyS25}, which we restate here for completeness. We simply bound together the case of small and large $\sigma$.

\begin{lemma}[{\cite[Lemma 2]{DBLP:conf/pqcrypto/PoulyS25}}]
    Consider a real number $\sigma>0$ and define 
    \begin{equation*}
        \tilde \rho(\sigma) = \begin{cases}
            1+2\exp{-\pi/\sigma^2}\text{ if }\sigma\leqslant 1\\
            \sigma(1+2\exp{-\pi\sigma^2})\text{ otherwise}.
        \end{cases}
    \end{equation*}
    Then $\tilde \rho$ is a continuously increasing function, and for any $\sigma>0$, 
    \begin{equation*}
        0<\rho_\sigma(\Z)-\tilde \rho(\sigma) \leqslant 2\sum_{k=2}^\infty \exp{-\pi k^2}\leqslant \gamma := 6.974685811\times 10^{-6}
    \end{equation*}
\end{lemma}

\subsection{Proof of~\autoref{lemma: inequality on gaussian masses on finite sets}}

In order to prove~\autoref{lemma: inequality on gaussian masses on finite sets}, we use the following. 

\begin{lemma}[{\cite[Lemma 2.10]{DBLP:journals/siamcomp/MicciancioR07}}]\label{lemma:rho sigma c on a finite subset}
    Let $\L$ be a full rank lattice in $\R^m$, and $\sigma>0$, $t>1/\sqrt{2\pi}$, and $\vec c\in \R^m$. Let $\mathcal B$ be the ball of $\R^m$ of center $0$ and radius 1, according to the euclidean distance. Then, we have 
    \begin{gather*}
        \rho_{\sigma}(\L\backslash t\sqrt{m}\mathcal B) < C^m \rho_{\sigma}(\L)\\
        \rho_{\sigma,\vec c}(\L\backslash t\sqrt{m}\mathcal B) < 2C^m \rho_{\sigma}(\L)
    \end{gather*}
    where $C := t\sqrt{2\pi e}\times \exp{-\pi t^2}$.
\end{lemma}

\gaussianfinitesets*

\begin{proof}
    Let us prove the first inequality. By~\autoref{lemma:rho sigma c}, one has
    \begin{equation*}
        \rho_{\sigma,\vec c}(\iintv{-2^M,2^M})\leqslant \rho_{\sigma,\vec c}(\Z) \leqslant \rho_{\sigma}(\Z).
    \end{equation*}
    Note that, since $M>0$, we have $2^M>1/\sqrt{2\pi}$. Then, we can apply~\autoref{lemma:rho sigma c on a finite subset} with $t = 2^M$ and $m = 1$, which gives:
    \begin{align*}
        \rho_{\sigma}(\Z) &= \rho_{\sigma}(\iintv{-2^M,2^M}) + \rho_{\sigma}(\Z\backslash \iintv{-2^M,2^M})\\
        &\leqslant \rho_{\sigma}(\iintv{-2^M,2^M}) + C \rho_{\sigma}(\Z)\enspace,
    \end{align*}
    where $C := 2^M\sqrt{2\pi e}\times \exp{-\pi 2^{2M}} < 1$.  Therefore, we have
    \begin{align*}
    	\rho_{\sigma}(\Z) &\leqslant \rho_{\sigma}(\iintv{-2^M,2^M})\times \frac{1}{1-C}\enspace\\
    	&=\rho_{\sigma}(\iintv{-2^M,2^M})\times \frac{1}{1-2^M\sqrt{2\pi e}\times \exp{-\pi 2^{2M}}}\enspace.
    \end{align*}
    Let us then prove the second inequality. By~\autoref{lemma:rho sigma c on a finite subset} with $t = 2^M$, we also have:
    \begin{equation*}
        \rho_{\sigma, \vec c}(\Z\backslash\iintv{-2^M,2^M})\leqslant 2C \rho_\sigma(\Z)\enspace.
    \end{equation*}
    But then, by~\autoref{lemma:rho Z approx}, we have $\rho_\sigma(\Z)\leqslant 3\sigma + 4$, therefore,
    \begin{multline*}
        \rho_{\sigma,\vec c}(\Z\backslash\iintv{-2^M,2^M})\leqslant 2^{M+1}\sqrt{2\pi e}\times \exp{-\pi 2^{2M}}( 3\sigma + 4 )  = O(\sigma\times 2^M \exp{-\pi 2^{2M}})\enspace.
    \end{multline*}
    This concludes the proof of the second inequality. \qed
\end{proof}

\subsection{Proof of~\autoref{prop: bound on some dtv and gaussian mass}}

\somedtvgaussianmass*


\begin{proof}

Note $B = QR$ the QR decomposition of $B$. The matrix $Q$ is orthonormal, $R$ is upper triangular, and each diagonal coefficient $R_{i,i}$ of $R$ equals $||\widetilde{\vec b}_i||$, where the $\widetilde{\vec b}_i$'s are the columns of the Gram-Schmidt transform of $B$. Then, we have 
\begin{equation*}
\rho_{\sigma}(\L\backslash \Omega) \leqslant \!\!\! \sum_{\vec x\in \Z^m\backslash \prod_{i=1}^m \iintv{-2^{M-1},2^{M-1}}} \!\!\!\!\!  \exp{-\frac{\pi}{\sigma^2}||B\vec x||^2} = \!\!\! \sum_{\vec x\in \Z^m\backslash \prod_{i=1}^m \iintv{-2^{M-1},2^{M-1}}}  \!\!\!\!\! \exp{-\frac{\pi}{\sigma^2}||R\vec x||^2}.
\end{equation*}
Since $R$ is upper triangular, it is straightforward to see that its eigenvalues are precisely the $R_{i,i} = ||\widetilde{\vec b}_i||$. We then have:
\begin{align*}
\rho_{\sigma}(\L\backslash \Omega) & \leqslant \sum_{\vec x\in \Z^m\backslash \prod_{i=1}^m \iintv{-2^{M-1},2^{M-1}}}\exp{-\frac{\pi\min(||\widetilde{\vec b}_i||)^2}{\sigma^2}||\vec x||^2} \\
& = \rho_{\sigma/\min(||\widetilde{\vec b}_i||)} \left( \Z^m\backslash \prod_{i=1}^m\iintv{-2^{M-1},2^{M-1}} \right) \enspace.
\end{align*}
We notice that $\Z^m\backslash \prod_{i=1}^m\iintv{-2^{M-1},2^{M-1}} = \left( \Z \backslash \iintv{-2^{M-1},2^{M-1}} \right)^m$. We can then apply~\autoref{lemma:rho sigma c on a finite subset} in dimension 1, and obtain:
\begin{equation*}
\rho_{\sigma}(\L\backslash \Omega) \leqslant \left(2^{M-1}\sqrt{2\pi e}\right)^m\times \exp{-\pi\cdot m\cdot 2^{2M-2}} (\rho_{\sigma/\min(||\widetilde{\vec b}_i||)}(\Z))^m \enspace.
\end{equation*}
Finally, by~\autoref{lemma:rho Z approx}, we have $\rho_{\sigma/\min(||\widetilde{\vec b}_i||)}(\Z) \leqslant 3(\sigma/\min(||\widetilde{\vec b}_i||)) + 4$, so we conclude.


This concludes the proof. \qed
\end{proof}

\subsection{Proof of~\autoref{lemma: third distance bound between gw and f L,sigma}}

\lembound*

\begin{proof}
    Let us set $X =\{0,\dotsc,q-1\}^m$. Then, by~\autoref{corollary:montanaro bis}, with probability at least $1-q^m\delta$, we have $|\tilde f(\vec x) - f_{\L_q(B),1/\sigma}(\vec x)|\leqslant \delta + \eta_1$, for all $\vec x\in X$. Additionally, $\L = \L_q(B)$ is a $q$-ary lattice, which means $q\Z^m\subset \L$. So, $\L+X\supset q\Z^m+\{0,\dotsc,q-1\}^m = \Z^m$. Therefore, this means that with probability at least $1-q^m\delta$, we have $|\tilde f(\vec x) - f_{\L_q(B),1/\sigma}(\vec x)|\leqslant \epsilon + \eta_1$, for all $\vec x\in \Z^m$.

    By~\autoref{lemma:distance from a lattice estimation}, we have
    \begin{equation*}
        \tilde f(\vec e )\geqslant f_{\L_q(B),1/\sigma}(\vec e)-\delta-\eta_1  \geqslant \rho_{1/\sigma}(\vec e)-\epsilon-\eta_1.
    \end{equation*}
    Let $\vec x\in \L\backslash \L_q(B)$, then $\vec z-\vec x\in \L+\L_q(B)$, and $\vec z-\vec x\neq 0$ for any $\vec z\in \L_q(B)$. As a result, $\L_q(B)-\vec x\subset (\L+\L_q(B))\backslash\{0\}$. Hence
    \begin{equation*}
        \mathrm{dist}(\vec x,\L_q(B)) = \min_{\vec z\in \L_q(B)}||\vec x+\vec z||\geqslant \min_{\vec y\in (\L+\L_q(B))\backslash \{0\}} ||\vec y|| = \lambda_1(\L+\L_q(B))\geqslant \tau+||\vec e||.
    \end{equation*}
    Then
    \begin{equation*}
        \mathrm{dist}(\vec e+\vec s,\L_q(B))\geqslant \mathrm{dist}(\vec x,\L_q(B))-||\vec e||\geqslant \tau.
    \end{equation*}
    We can apply~\autoref{lemma:distance from a lattice estimation} to get that for any $\vec x\in \L\backslash\{0\}$,
    \begin{equation*}
        \tilde f(\vec e+\vec x)\leqslant f_{\L_q(B),1/\sigma}(\vec e+\vec x)+\eta_1 + \epsilon  \leqslant \rho_{1/\sigma}(\mathrm{dist}(\vec e + \vec x,\L_q(B))-\tau)+\epsilon+\eta_1.
    \end{equation*}
    Since $\rho_{1/\sigma}:[0,\infty[\to \R$ is decreasing, and reusing the two previous inequalities, we have moreover
    \begin{align*}
        \rho_{1/\sigma}(\mathrm{dist}(\vec e+\vec x,\L_q(B))-\tau) &\leqslant \rho_{1/\sigma}(\mathrm{dist}(\vec x,\L_q(B))-||\vec e||-\tau)\\
        &\leqslant \rho_{1/\sigma}(\lambda_1(\L+\L_q(B))-||\vec e||-\tau).
    \end{align*}
    Putting everything together, we have
    \begin{equation*}
        \tilde f(\vec e)-\tilde f(\vec e+\vec x) \geqslant \rho_{1/\sigma}(\vec e)-\rho_{1/\sigma}(\lambda_1(\L+\L_q(B))-||\vec e||-\tau)-2\epsilon-2\eta_1 >\eta_2.
    \end{equation*}
    by assumption. \qed
\end{proof}

\end{document}